\documentclass[12pt]{article}
\usepackage{mydef}
\usepackage{mathtools}
\usepackage{booktabs}
\usepackage{wrapfig}
\usepackage{lscape}
\usepackage{rotating}
\usepackage[hide]{ed}
\usepackage{appendix}
\usepackage{authblk}

\newenvironment{proof}[1][Proof]{\begin{trivlist}
\item[\hskip \labelsep {\bfseries #1}]}{\end{trivlist}}

\begin{document}
\title{LSV models with stochastic interest rates and correlated jumps}
\author{Andrey Itkin}
\affil{\small New York University, Tandon School of Engineering}

\date{\today}

\maketitle

\begin{abstract}
Pricing and hedging exotic options using local stochastic volatility models drew a serious attention within the last decade, and nowadays became almost a standard approach to this problem. In this paper we show how this framework could be extended by adding to the model stochastic interest rates and correlated jumps in all three components. We also propose a new fully implicit modification of the popular Hundsdorfer and Verwer and Modified Craig-Sneyd finite-difference schemes which provides second order approximation in space and time, is unconditionally stable and preserves positivity of the solution, while still has a linear complexity in the number of grid nodes.
\end{abstract}

Pricing and hedging exotic options using local stochastic volatility (LSV) models drew a serious attention within the last decade, and nowadays became almost a standard approach to this problem. For the detailed introduction into the LSV among multiple available references we mention a recent comprehensive literature overview in \cite{Homescu2014}. Note, that the same model or its flavors appear in the literature under different names, such as stochastic local volatility model, universal volatility model of \cite{Lipton2002}, unspanned stochastic local volatility model (USLV) of \cite{HalperinItkin2013}, etc.

Despite LSV has a lot of attractive features allowing simultaneous pricing and calibration of both vanilla and exotic options, it was observed that in many situations, e.g., for short maturities, jumps in both the spot price and the instantaneous variance need to be taken into account to get a better replication of the market data on equity or FX derivatives. This approach was pioneered by \cite{Bates:1996b} who extended the Heston model by introducing jumps with finite activity into the spot price (a jump-diffusion model). Then \cite{Lipton2002} further extended this approach by considering local stochastic volatility to be incorporated into the jump-diffusion model (for the extension to an arbitrary L{\'e}vy model, see, e.g., \cite{Pagliarani2012}). Later  \cite{Sepp2011a,Sepp2011b} investigated exponential and discrete jumps in both the underlying spot price $S$ and the instantaneous variance $v$, and concluded that infrequent negative jumps in the latter are necessary to fit the market data on equity options\footnote{Here we don't discuss this conclusion. However, for the sake of reference note, that this could be dictated by some inflexibility of the Heston model where vol-of-vol is proportional to $v^{0.5}$. More flexible models which consider the vol-of-vol power to be parameter of calibration, \cite{Gatheral2008, Itkin2013}, might not need jumps in $v$. See also \cite{Sepp2014} and the discussion therein.}.
In \cite{Durham2013} a similar approach was proposed to use general jump-diffusion equations for modeling both $S$ and $v$.

Note, that in the literature jump-diffusion models for both $S$ and $v$ are also known under the name SVCJ (stochastic volatility with contemporaneous jumps). These models as applied to pricing American options were intensively studied in \cite{Toivanen2014}, for basket options in \cite{Shirava2013}.

Another way to extend the LSV model is to assume that the short interest rates $r$ could be stochastic. Under this approach jumps are ignored, but instead a system of three Stochastic Differential Equations (SDE) with drifts and correlated diffusions is considered, see \cite{Giese2006,Medvedev2010,Oosterlee2011,Hout3D,Hilpisch2011,Chiarella2013,BL2013} and references therein.

As we have already mentioned, accounting for jumps could be important to calibrate the LSV model to the market data. And making the interest rate stochastic doesn't violate this conclusion. Moreover, jumps in the interest rate itself could be important. For instance, in \cite{Chen2004} a stochastic volatility model with jumps in both rates and volatility was calibrated to the daily data for futures interest rates in four major currencies which provided a better fit for the empirical distributions. Also the results in \cite{Johannes2004} obtained using Treasury bill rates find evidence for the presence of jumps which play an important statistical role. Also in \cite{Johannes2004} it was found that jumps generally have a minor impact on yields, but they are important for pricing interest rate options.

In the FX world there exist some variations of the discussed models. For instance, in \cite{Doffou2001} foreign and domestic interest rates are stochastic with no jumps while the exchange rate is modeled by jump-diffusion. In \cite{cw04} both domestic and foreign rates were represented as a L{\'e}vy process with the diffusion component using a time-change approach. The diffusion components could be correlated in contrast to the jump components.

In the bond market, as shown in \cite{Das2002}, the information surprises result in discontinuous interest rates. In that paper a class of Poisson--Gaussian models of the Fed Funds rate was developed to capture the surprise effects. It was shown that these models offer a good statistical description of a short rate behavior, and are useful in understanding many empirical phenomena. Jump (Poisson) processes capture empirical features of the data which would not be captured by Gaussian models. Also there is strong evidence that the existing Gaussian models would be well-enhanced by jump and ARCH-type processes.

Overall, it would be desirable to have a model where the LSV framework could be combined with stochastic rates and jumps in all three stochastic drivers. We also want to treat these jumps as general L{\'e}vy processes, so not limiting us by only the jump-diffusion models. In addition, we consider Brownian components to be correlated as well as the jumps in all stochastic drivers to be correlated, while the diffusion and jumps remain uncorrelated. Finally, since such a model is hardly analytically tractable when parameters of the model are time-dependent (which is usually helpful to better calibrate the model to a set of instruments with different maturities, or to a term-structure of some instrument), we need an efficient numerical method for pricing and calibration.

For this purpose in this paper we propose to exploit our approach first elaborated on in \cite{ItkinLipton2014} for modeling credit derivatives. In particular, in the former paper we considered a set of banks with mutual interbank liabilities whose assets are driven by correlated L{\'e}vy processes. For every asset, the jumps were represented as a weighted sum of the common and idiosyncratic parts. Both parts could be simulated by an arbitrary L{\'e}vy model which is an extension of the previous approaches where either the discrete or exponential jumps were considered, or a L{\'e}vy copula approach was utilized. We provided a novel efficient (linear complexity in each dimension) numerical (splitting) algorithm for solving the corresponding 2D and 3D jump-diffusion equations, and proved its convergence and second order of accuracy in both space and time. Test examples were given for the Kou model, while the approach is in no way limited by this model.

In this paper we demonstrate how a similar approach can be used together with the Metzler model introduced by \cite{Schoutens98, Schoutens01}. It is built based on the Meixner distribution which belongs to the class of the infinitely divisible distributions. Therefore, it gives rise to a L{\'e}vy process - the Meixner process. The Meixner process is flexible and analytically tractable, i.e. its pdf and CF are known in closed form (in more detail see, e.g., \cite{Itkin2014a} and references therein). The Meixner model is known to be rich and capable to be calibrated to the market data. Again, this model is chosen only as an example, because, in general, the approach in use is rather universal.

We also propose a new fully implicit modification of the popular Hundsdorfer and Verwer and Modified Craig-Sneyd finite-difference schemes which provides second order approximation in space and time, is unconditionally stable and preserves positivity of the solution, while still has a linear complexity in the number of grid nodes. This modification allows elimination of first few Rannacher steps as this is usually done in the literature to provide a better stability (see survey, e.g., in \cite{Hout3D}), and provides much better stability of the whole scheme which is important when solving multidimensional problems.

The rest of the paper is organized as follows. In the next section we describe the model. Section~\ref{SolSect} consists of two subsections. The first one introduces the new splitting method, which treats mixed derivatives terms implicitly, thus providing a much better stability. The second subsection describes how to deal with jumps if one uses the Meixner model. However, by no means this approach is restricted just by this model as, e.g., in \cite{ItkinLipton2014} we used the Kou jump models using the same treatment of the jump terms. So here the Meixner model is taken as another example. Section~\ref{numExamp} presents the results of some numerical experiments where prices of European vanilla and barrier options were computed using these model and numerical method. The final section concludes.

\section{Model}
We consider an LSV model with stochastic interest rates and jumps by introducing stochastic dynamics for variables $S_t,v_t,r_t$. We assume that it could include both diffusion and jumps components, as follows:
\begin{align} \label{sde}
d S_t &= (r_t - q)S_t dt + \sigma_{s}(S_t,t) S_t^c\sqrt{v_t} W_s + S_t d L_{S_t,t}, \\
d v_t &= \kappa_v(t)[\theta_v(t) - v_t] dt + \xi_v v_t^a W_v + v_t dL_{v_t,t}, \nonumber \\
d r_t &= \kappa_r(t)(\theta_r(t) - r_t) dt + \xi_r r_t^b W_r + r_t dL_{r_t,t}. \nonumber
\end{align}
Here $q$ is the continuous dividend, $t$ is the time, $\sigma_s$ is the local volatility function, $W_s, W_v, W_r$ are correlated Brownian motions, such that $<dW_i,dW_j>=\rho_{ij}dt, \ i,j \in [s,v,r]$, $\kappa_v, \theta_v, \xi_v$ are the mean-reversion rate, mean-reversion level and volatility of volatility (vol-of-vol) for the instantaneous variance $v_t, \kappa_r, \theta_r, \xi_r$ are the corresponding parameters for the stochastic interest rate $r_t$, $0 \le a < 2, \ 0 \le b < 2, \ 0 \le c < 2$ are some power constants which are introduced to add additional flexibility to the model as compared with the popular Heston ($a = 0.5$), lognormal ($a = 1$) and 3/2 ($a = 1.5$)
models\footnote{If, however, somebody wants to determine these parameters by calibration, she has to be careful, because having both vol-of-vol and a power constant in the same diffusion term brings an ambiguity into the calibration procedure. Nevertheless, this ambiguity can be resolved if for calibration some additional financial instruments are used, e.g., exotic option prices are combined with the variance swaps prices, see \cite{Itkin2013}.}. Processes $L_s, L_v, L_r$ are pure discontinuous jump processes with generator $A$
\[
A f(x) = \int_{\mathbb{R}}\left( f(x+y) - f(x) - y \mathbf{1}_{|y|<1}\right) \mu (dy),
\]
\noindent with $\mu (dy)$ be a L{\'e}vy measure, and
\[
\int_{|y|>1}e^{y}\mu (dy)<\infty.
\]
At this stage, the jump measure $\mu (dx)$ is left unspecified, so all types of jumps including those with finite and infinite variation, and finite and infinite activity could be considered here.

Following \cite{ItkinLipton2014} we introduce correlations between all jumps as this was done in \cite{Ballotta2014}. They construct the jump process as a linear combination of two independent L{\'e}vy processes representing the systematic factor and the idiosyncratic shock, respectively (see also \cite{ContTankov}). It has an intuitive economic interpretation and retains nice tractability, as the multivariate characteristic function in this model is available in closed form based on the following proposition of \cite{Ballotta2014}:
\begin{proposition} \label{propBallotta}
\label{prop0} Let $Z_t, \ Y_{j,t}, \ j=1,...,n$ be independent L{\'e}vy
processes on a probability space $(Q, F, P )$, with characteristic functions
$\phi_Z(u; t)$ and $\phi_{Y_j}(u; t)$, for $j=1,...,n$ respectively. Then,
for $b_j \in \mathbb{R}, \ j=1,...,n$
\[
X_t = (X_{1,t},..., X_{n,t})^\top = (Y_{1,t} + b_1 Z_t,...,Y_{n,t} + b_n
Z_t)^\top
\]
is a L{\'e}vy process on $\mathbb{R}^n$. The resulting characteristic function
is
\[
\phi_{\mathbf{X}}(\mathbf{u}; t) = \phi_Z\left( \sum_{i=1}^n b_i u_i;
t\right) \prod_{i=1}^n \phi_{Y_j}(u_j;t), \qquad \mathbf{u} \in \mathbb{R}^n.
\]
\end{proposition}

By construction every factor $X_{i,t}, i=1,...,n$ includes a common factor $Z_t$. Therefore, all components $X_{i,t}, i=1,...,n$ could jump together, and loading factors $b_i$ determine the magnitude (intensity) of the jump in $X_{i,t}$ due to the jump in $Z_t$. Thus, all components of the multivariate
L{\'e}vy process $\mathbf{X}_t$ are dependent, and their pairwise correlation is given by (again
see \cite{Ballotta2014} and references therein)
\[
\rho_{j,i} = \mbox{Corr}(X_{j,t}, X_{i,t}) = \dfrac{b_j b_i \mbox{Var}(Z_1)} {\sqrt{\mbox{Var}(X_{j,1})}\sqrt{\mbox{Var}(X_{i,1})}}.
\]
Such a construction has multiple advantages, namely:

\begin{enumerate}
\item As $\mbox{sign}(\rho_{i,j}) = \mbox{sign}(b_i b_j)$, both positive and
negative correlations can be accommodated

\item In the limiting case $b_i \to 0$ or $b_j \to 0$ or $\mbox{Var}(Z_1) =
0 $ the margins become independent, and $\rho_{i,j} = 0$. The other limit $b_i \to \infty$ or $b_j \to \infty$ represents a full positive correlation
case, so $\rho_{i,j} = 1$. Accordingly, $b_i \to \infty, \ b_{3-i} \to
\infty, \ i=1,2$ represents a full negative correlation case as in this
limit $\rho_{i,j} = -1$.
\end{enumerate}
According to this setup, the total instantaneous correlation between the assets $x_i$ and $x_j$ reads \cite{Ballotta2014}
\begin{equation}  \label{corr}
\tilde{\rho}_{ij} = \dfrac{\rho \sigma_i \sigma_j + b_i b_j \mbox{Var}(Z_1)}{\sqrt{\sigma_i^2 + \mbox{Var}(X_{i,1})} \sqrt{\sigma_j^2 + \mbox{Var}(X_{j,1})}}.
\end{equation}

To price contingent claims, e.g., vanilla or exotic options written on the underlying spot price, by using a standard technique as in \cite{ContTankov}, the following multidimensional PIDE could be derived which describes the evolution of the option price $V$ under risk-neutral measure
\begin{equation} \label{PIDE}
V_{\tau }=[\mathcal{D}+\mathcal{J}]V,
\end{equation}
\noindent where $\tau = T - t$ is the backward time, $T$ is the time to the contract expiration, $\mathcal{D}$ is the three-dimensional linear convection-diffusion operator of the form
\begin{align}  \label{Dif}
\mathcal{D} &= F_0 + F_1 + F_2 + F_3, \\
F_1 &= (r - q) S \fp{}{S} + \dfrac{1}{2} \sigma_s^2 S^{2c} v \sop{}{S} - \frac{1}{3}r, \nonumber \\
F_2 &= \kappa_v(t)[\theta_v(t) - v] \fp{}{v} + \dfrac{1}{2} \xi_v^2 v^{2 a} \sop{}{v}  - \frac{1}{3}r, \nonumber \\
F_3 &= \kappa_r(t)[\theta_r(t) - r] \fp{}{r} + \dfrac{1}{2} \xi_r^2 r^{2 b} \sop{}{r}  - \frac{1}{3}r, \nonumber \\
F_0 &= \tilde\rho_{s,v} \sqrt{\Var(S_t) \Var(v_t)} \cp{}{S}{v} + \tilde\rho_{s,r} \sqrt{\Var(S_t) \Var(R_t)} \cp{}{S}{R} \nonumber \\
&+  \tilde\rho_{v,r} \sqrt{\Var(v_t) \Var(R_t)}\cp{}{R}{v}, \nonumber \\
\Var(S_t) &= \sigma^2_s(S,t)S^{2c} v + S^2 (\Var(Y_{s,1}) + b_s^2 \Var(Z(1))), \nonumber \\
\Var(v_t) &= \xi^2_v v^{2a} + v^2(\Var(Y_{v,1}) + b_v^2 \Var(Z(1))), \nonumber \\
\Var(r_t) &= \xi_r^2 r^{2b} + r^2(\Var(Y_{r,1}) + b_r^2 \Var(Z(1))), \nonumber
\end{align}
\noindent and $\mathcal{J}$ is the jump operator
\begin{align}  \label{Jump}
\mathcal{J}V &= \int_{-\infty}^{\infty}\Big[V(x_s + y_s, x_v + y_v, x_r + y_r, \tau)
- V(x_s,x_v,x_r, \tau)  \\
&- \sum_{\chi \in [s,v,r]} (e^{y_\chi} - 1)\fp{V(x_s,x_v,x_r,\tau)}{\chi}\Big] \mu(d y_s d y_v dy_r), \nonumber
\end{align}
\noindent where $\mu(d y_s d y_v dy_r)$ is the three-dimensional L{\'e}vy measure, and $x_s = \log S/S_0, \ x_v = \log v/v_0, \ x_r = \log r/r_0$.

This PIDE has to be solved subject to the boundary and terminal conditions.
We assume that the terminal condition for equity derivatives reads
\[ V(S,v,r,T) = P(S), \]
\noindent where $P(S)$ is the option payoff as defined by the corresponding contract. The boundary conditions could be set, e.g., as in \cite{Hout3D}. However, in the presence of jumps these conditions should be extended as follows. Suppose we want to use finite-difference method to solve the above PIDE and construct a jump grid, which is a superset of the finite-difference grid used to solve the diffusion equation (i.e. when $\mathcal{J} = 0$, see \cite{Itkin2014}). Then these boundary conditions should be set on this jump grid as well as at the boundaries of the diffusion domain.

\section{Solution of the PIDE} \label{SolSect}
To solve \eqref{PIDE} we use a splitting algorithm described in \cite{ItkinLipton2014}. The algorithm provides the second order approximation in time (assuming that at every splitting step the corresponding problem is solved with the same order of approximation) and reads
\begin{align} \label{splitting}
V(\tau + \Delta \tau) &=
e^{0.5 \Delta \tau \mathcal{D}}
e^{0.5 \Delta \tau \mathcal{J}_s}
e^{0.5 \Delta \tau \mathcal{J}_v}
e^{0.5 \Delta \tau \mathcal{J}_r}
e^{\Delta \tau \mathcal{J}_{123}} \\
&\cdot e^{0.5 \Delta \tau \mathcal{J}_r}
e^{0.5 \Delta \tau \mathcal{J}_v}
e^{0.5 \Delta \tau \mathcal{J}_s}
e^{0.5 \Delta \tau \mathcal{D}}V(\tau), \nonumber \\
J_\chi &= \phi_\chi(- i \triangledown_\chi), \ J_{123} = \phi_Z(- i \sum_{\chi \in [x,v,r]} b_\chi \triangledown_\chi), \ \triangledown_\chi \equiv \fp{}{\chi}. \nonumber
\end{align}
Thus, this requires a sequential solution of 9 equations at every time step. The first and the last steps are pure  advection-diffusion problems and could be solved using, e.g., a finite difference method proposed in \cite{Hout3D}. We, however, slightly modified it by replacing an explicit scheme for the mixed derivative operators with the implicit ones. The detailed description of this approach as well as our reasons for doing that are given in the next section.

\subsection{Advection-diffusion problem}

We follow \cite{HoutWelfert2007}, who consider the unconditional stability of second-order finite-difference schemes used to numerically solve multi-dimensional diffusion problems containing mixed spatial derivatives. They investigate the ADI scheme proposed by Craig and Sneyd (see references in the paper), a modified version of Craig and Sneyd's ADI scheme, and the ADI scheme introduced by Hundsdorfer and Verwer. The necessary conditions are derived on the parameters of each of these schemes for unconditional stability in the presence of mixed derivative terms.

For example, let us choose a HV scheme. The main result of \cite{HoutWelfert2007} is that under some mild conditions on the parameter $\theta$ of the scheme, the second-order Hundsdorfer and Verwer (HV) scheme is unconditionally stable when applied to semi-discretized diffusion problems with mixed derivative terms in an arbitrary spatial dimension $k > 2$. For the 3D convection-diffusion problem in \eqref{PIDE} with $\mathcal{J} = 0$, the HV scheme defines an approximation $V_n \approx V(\tau_n)$, $n=1,2,3, \dots$, by performing a series of (fractional) steps:
\begin{align} \label{HV}
Y_0 &= V_{n-1} + \Delta \tau F(\tau_{n-1})V_{n-1}, \\
Y_j &=  Y_{j-1} + \theta \Delta \tau \left[F_j(\tau_n)Y_j - F_j(\tau_{n-1})V_{n-1}\right], \ j = 1,2,3 \dots,k, \nn \\
\tilde{Y}_0 &= Y_0 + \dfrac{1}{2} \Delta \tau \left[F(\tau_n) Y_k - F(\tau_{n-1})V_{n-1}\right], \nn \\
\tilde{Y}_j &= \tilde{Y}_{j-1} + \theta \Delta \tau \left[F_j(\tau_n)\tilde{Y}_j - F_j(\tau_n) Y_k\right], \ j=1,2,3 \dots,k, \nn \\
V_n &= \tilde{Y}_k, \nn
\end{align}
\noindent where $F = \sum_j F_j, \ j=0,1...k$. This scheme is of order two in time  for any value of $\theta$, so this parameter can be chosen to meet additional requirements, \cite{HoutWelfert2007}. An advantage of this scheme is that the fractional steps in line 1 and 3 in \eqref{HV}, which deal with the mixed derivatives, are solved by using an explicit scheme. At the same time this could bring a problem, because a very careful approximation of the mixed derivative term is required to preserve the stability and positivity of the solution\footnote{This is especially important at the first few steps in time because of a step-function nature of the payoff. So a smoothing scheme, e.g., \cite{rannacher}, is usually applied at the first steps, which, however, loses the second order approximation at these steps.}. Sometimes this requires very small step in time to be chosen.

In the 2D case to resolve this rather delicate issue a seven-point stencil for discretization of the mixed-derivative operator that preserves the positivity of the solution was proposed in \cite{toivanen2010, chiarella2008} for correlations $\rho < 0$, and in \cite{IkonenToivanen2008, IkonenToivanen2007} for positive correlations. However, in their schemes the mixed derivative term was treated implicitly (that is the reason they needed a discretized matrix to be an M-matrix). In our case the entire matrix in the right-hand side of steps 3,5 should be either a positive matrix, or a Metzler matrix (in this case the negative of an M-matrix). The latter can be achieved when using approximations of \cite{toivanen2010, chiarella2008} and \cite{IkonenToivanen2008, IkonenToivanen2007} in an opposite order, i.e. use approximations recommended for $\rho > 0$ when $\rho < 0$, and vice versa. However, due to the nature of the 7-point stencil, they are not able to provide a rigorous second order approximation of the mixed derivatives.

In our numerical experiments even using these explicit analogs of the mixed derivatives approximation in the 3D case was not always sufficient. Indeed, either we use real second order approximation of the mixed derivatives relying on the fact that in the HV splitting scheme $F_0$ comes only as a part of $F$. Hence, the negative terms in $F_0$ can be partly or even fully compensated by the other terms. Unfortunately, at some values of the model parameters this could be insufficient to provide the total positivity of the solution. Or we use the 7-points stencil which works well for the implicit scheme (for the reasons which will became clear in Appendix~\ref{appPos} when proving our Theorem), but still doesn't provide a necessary stability for the explicit scheme. Thus, one has to choose a very small step in time, which is impractical.
Therefore, in this paper to provide an additional stability of the whole splitting scheme we modified this step as follows.

The main idea is to sacrifice the simplicity of the explicit representation of the mixed derivative term for the better stability. That is what was done in \cite{toivanen2010,chiarella2008,IkonenToivanen2008, IkonenToivanen2007} who dealt with a 2D case and used an implicit approximation of the mixed derivatives term. However, in this paper we propose another approach.

Below, for the sake of simplicity of the explanation, suppose when solving \eqref{HV} we don't take into account jumps. This assumption could be easily relaxed since jumps add just an additional term into the definition of $\Var$ in \eqref{Dif}. Consider the first step in \eqref{HV}. Since here only the first order approximation in time is necessary, this step can be re-written in two steps
\begin{align} \label{twoStep}
\fp{V^*}{\tau} &= F_0(\tau_{n-1})
\equiv F_{Sv}(\tau_{n-1}) + F_{Sr}(\tau_{n-1}) + F_{vr}(\tau_{n-1}), \\
V(\tau) &= V^*(\tau) + \Delta \tau [F_1(\tau_{n-1}) + F_2(\tau_{n-1}) + F_3(\tau_{n-1})]V^*(\tau), \nonumber
\end{align}
\noindent with
\begin{align} \label{defF0}
F_{Sv}(\tau_{n-1}) &= \rho_{s,v} \sigma_s(S,\tau) S^c \xi_v(\tau) v^{a+0.5} \cp{}{S}{v} \equiv
\rho_{s,v} W(S) W(v) \cp{}{S}{v},  \nonumber \\
F_{Sr}(\tau_{n-1}) &=  \rho_{s,r} \sigma_s(S,\tau)S^c \sqrt{v} \xi_r(\tau) r^b \cp{}{S}{r} \equiv
\rho_{s,r} W(S) W(r)\sqrt{v} \cp{}{S}{r}, \nonumber \\
F_{vr}(\tau_{n-1}) &= \rho_{v,r} \xi_v(\tau) v^{a} \xi_r(\tau) r^b \cp{}{r}{v} \equiv
\rho_{v,r} \dfrac{W(v) W(r)}{\sqrt{v}}  \cp{}{r}{v}, \nonumber
\end{align}
\noindent where $W(S) = \sigma(S,\tau)S^c, \ W(v) = \xi_v(\tau) v^{a+0.5}, \ W(r) = \xi_r(\tau) r^{b}$.

So efficiently at the first sub-splitting step we take a liberty to solve the first equation of \eqref{twoStep} as we like, and the remaining part (the second sub-step) is treated explicitly, e.g. in the same way as in the HV scheme.

Now, a general solution of this first equation in \eqref{twoStep} can be written in the operator form as
\begin{equation*}
V(\tau + \Delta \tau) = e^{\Delta \tau (F_{Sv}(\tau_{n-1}) + F_{Sr}(\tau_{n-1}) + F_{vr}(\tau_{n-1}))}V(\tau).
\end{equation*}
Again, with $O(\dtau)$
\begin{equation*} \label{matrixSol}
V(\tau + \Delta \tau) = e^{\Delta \tau F_{Sv}(\tau_{n-1})} e^{\Delta \tau F_{Sr}(\tau_{n-1})} e^{\Delta \tau F_{vr}(\tau_{n-1})}V(\tau),
\end{equation*}
\noindent or using splitting
\begin{align} \label{expSpl}
V^{(1)} &= e^{\Delta \tau F_{Sv}(\tau_{n-1})}V(\tau), \\
V^{(2)} &= e^{\Delta \tau F_{Sr}(\tau_{n-1})}V^{(1)}, \nonumber \\
V(\tau + \Delta \tau) &= e^{\Delta \tau F_{vr}(\tau_{n-1})}V^{(2)}. \nonumber
\end{align}
The order of the splitting steps usually doesn't matter.

Accordingly, it is sufficient to consider just one step in \eqref{expSpl} since the others can be done in the similar way. For example, below let us consider step 1. First, we use \pade approximation (0,1) which provides approximation of the first line in \eqref{expSpl} with the first order in $\Delta \tau$, and is implicit. Approximation wise this is equivalent to the first line of \eqref{HV}. Having that,
the first equation in \eqref{expSpl} transforms to
\begin{equation} \label{trick1}
\left[1 - \Delta \tau \rho_{s,v} W(S) W(v) \triangledown_S \triangledown_v \right] V^{(1)} = V(\tau).
\end{equation}
Second, we again rewrite it using a trick\footnote{The trick is motivated by the desire to build an ADI scheme which consists of two one-dimensional steps, because for the 1D equations we know how to make the rhs matrix to be an EM-matrix \cite{Itkin2014}.}
\begin{align} \label{trickNeg}
\Big(P &- \sdt \rho_{s,v} W(S)\triangledown_S \Big)\Big(Q + \sdt W(v) \triangledown_v\Big) V^{(1)}  \\
	&= V(\tau) + \left[(P Q-1) - Q\sdt \rho_{s,v} W(S)\triangledown_S + P \sdt W(v) \triangledown_v  \right]V^{(1)}, \nonumber
\end{align}
\noindent where $P,Q,$ are some positive numbers which have to be chosen based on some conditions, e.g., to provide diagonal dominance of the matrices in the parentheses in the lhs of \eqref{trickNeg}, see below.

The intuition for this representation is as follows. Suppose we need to solve some parabolic PDE and represent the solution in the form of a matrix exponential $V(\tau + \dtau) = e^{\dtau J}V(\tau)$. Since computing the matrix exponential might be expensive, to preserve the second order approximation in $\dtau$ one can use a second order \pade approximation. In this case, e.g., a popular Crank-Nicholson scheme preserves positivity of the solution only if the negative diagonal elements $d_{0,i}, \ i=1,...,N$ of $\frac{1}{2}\dtau J$ obey the condition $d_{0,i}(\frac{1}{2}\dtau J) + 1 > 0, \forall i \in (1,...,N)$. This efficiently issues some limitations on the time step $\dtau$. As a resolution, e.g., in \cite{WadeSmoothing} higher order fully implicit \pade approximations were proposed to be used instead of the Crank-Nicholson scheme. This solves the problem with getting a positive solution since
\[ e^{y} = 1 - y  + \frac{1}{2} y^2 + O(y^2) \approx \frac{1}{2}[y - (1+\imath)][y-(1-\imath)], \qquad y \equiv \dtau J, \quad \imath = \sqrt{-1}, \]
\noindent and by using an appropriate discretization each matrix in the parentheses can be made an M-matrix which inverse is a non-negative matrix. This can be done when $J$ is a 1D parabolic operator.
Performance-wise, this, however, gives rise to solving few (e.g., 2 in the case of \pade (0,2) approximation) systems of linear equations with complex numbers. Hence, the complexity of the solution is, at least, 4 times worse. Our representation \eqref{trickNeg} aims to utilize a similar idea, but being transformed to the iterative method. The key point here is that we use a theory of EM-matrices ( see \cite{Itkin2014, Itkin2014b} and references therein), and manage to propose a second order approximation of the first derivative which makes our matrices to be real EM-matrices. So again, the inverse of the latter is a positive matrix.

\eqref{trickNeg} can be solved using fixed-point Picard iterations. One can start with setting $ V(1) = V^0 = V(\tau)$ in the rhs of \eqref{trickNeg}, then solve sequentially two systems of equations
\begin{align} \label{spl2}
\Big(Q &+ \sdt W(v) \triangledown_v\Big) V^*
	= V(\tau) + \Big[P Q -1  - Q\sdt \rho_{s,v} W(S)\triangledown_S
	+ P \sdt W(v)\triangledown_v  \Big]V^k, \nonumber \\
\Big(P &- \sdt \rho_{s,v} W(S) \triangledown_S \Big)V^{k+1} = V^*.
\end{align}
Here $V^k$ is the value of $V^{(1)}$ at the $k$-th iteration.


Before constructing a finite difference scheme to solve this equation
we need to introduce some definitions. Define a one-sided {\it forward} discretization of $\triangledown$, which we denote as $A^F_x: A^F_x C(x) = [C(x+h,t) - C(x,t)]/h$. Also define a one-sided {\it backward} discretization of $\triangledown$, denoted as $A^B_x: \ A^B_x C(x) = [C(x,t) - C(x-h,t)]/h$. These discretizations provide first order approximation in $h$, e.g., $\triangledown C(x) = A^F_x C(x) +  O(h)$.
To provide the second order approximations, use the following definitions. Define $A^C_{2,x} = A^F_x  A^B_x$ - the {\it central} difference approximation of the second derivative $\triangledown^2_x$, and $A^C_x = (A^F_x + A^B_x)/2$ - the {\it central} difference approximation of the first derivative $\triangledown$. Also define a one-sided second order approximations to the first derivatives: {\it backward} approximation  $A^B_{2,x}: \ A^B_{2,x} C(x) = [ 3 C(x) - 4 C(x-h) + C(x-2h)]/(2 h)$, and {\it forward} approximation $A^F_{2,x}: \ A^F_{2,x} C(x) = [ -3 C(x) + 4 C(x+h) - C(x+2h)]/(2 h)$. Also $I_x$ denotes a unit matrix.
All these definitions assume that we work on a uniform grid, however this could be easily generalized for the non-uniform grid as well, see, e.g., \cite{HoutFoulon2010}.

For solving \eqref{spl2} we propose two FD schemes. The first one (Scheme~A) is introduced by the following Propositions\footnote{For the sake of clearness we formulate this Proposition for the uniform grid, but it should be pretty much transparent how to extend it for the non-uniform grid.}:

\begin{proposition} \label{propPos}
Let us assume $\rho_{s,v} \ge 0$, and approximate the lhs of \eqref{spl2} using the following finite-difference scheme:
\begin{align} \label{fd1}
\Big(Q I_v + \sdt W(v) A^B_{2,v}\Big) V^*
	&= \alpha^{+}V(\tau) - V^k, \\
\Big(P I_S - \sdt \rho_{s,v} W(S) A^F_{2,S} \Big)V^{k+1} &= V^*, \nonumber	\\
\alpha^{+} = (PQ + 1)I &- Q\sdt \rho_{s,v} W(S) A_{S}^B + P\sdt W(v) A_{v}^F. \nonumber
\end{align}
Then this scheme is unconditionally stable in time step $\Delta \tau$, approximates \eqref{spl2} with $O(\sqrt{\dtau}\max(h_S,h_v))$ and preserves positivity of the vector $V(x,\tau)$ if $Q = \beta \sdt/h_v, \ P = \beta \sdt/h_S$, where $h_v, h_S$ are the grid space steps correspondingly in $v$ and $S$ directions, and the coefficient $\beta$ must be chosen to obey the condition:
\[
\beta > \max_{S,v}[W(v) + \rho_{s,v} W(S)].
\]
\end{proposition}
\begin{proof}
See Appendix~\ref{appPos}.
\end{proof}
The computational scheme in \eqref{fd1} should be understood in the following way. At the first line of \eqref{fd1} we begin with computing the product $V_1 = \alpha^+ V(\tau)$. This can be done in three steps. First, the product $V_1 = Q\sdt \rho_{s,v} W(S) A_{S}^B V(\tau)$ is computed in a loop on $v_i, i=1,...,N_2$. In other words, if $V(\tau)$ is a $N_1 \times N_2$ matrix where the rows represent the $S$ coordinate and the columns represent the $v$ coordinate, each $j$-th column of $V_1$ is a product of matrix $Q\sdt \rho_{s,v} W(S) A_{S}^B$ and the $j$-th column of $V(\tau)$. The second step is to compute the product $V_2 = P\sdt W(v) A_{v}^F V(\tau)$ which can be done in a loop on $S_i, \ i=1,...,N_1$. Finally, the right hand side of the first line in \eqref{fd1} is $(PQ+1) V(\tau) - V_1 + V_2 - V^k$.
Then in a loop on $S_i, \ i=1,...,N_1$, $N_1$ systems of linear equations have to be solved, each giving a row vector of $V^*$. The advantage of the representation \eqref{fd1} is that the product $\alpha^+ V(\tau)$ can be precomputed.

If $\sdt \approx \max(h_S,h_v)$, then the whole scheme becomes of the second order in space. However, this would be a serious restriction inherent to the explicit schemes. Therefore, in this paper we don't rely on it. Note, that in practice the time step is usually chosen such that $\sdt \ll 1$, and hence the whole scheme is expected to be closer to the second, rather than to the first order in $h$.
As it has been already mentioned, this is similar to \cite{toivanen2010, chiarella2008} for correlations $\rho < 0$, and \cite{IkonenToivanen2008, IkonenToivanen2007} for positive correlations, where a seven point stencil breaks a rigorous second order of approximation in space.

A similar Proposition can be proved in case $\rho_{s,v} \le 0$.
\begin{proposition} \label{propNeg}
Let us assume $\rho_{s,v} \le 0$, and approximate the lhs of \eqref{spl2} using the following finite-difference scheme of the second order in space:
\begin{align} \label{fd2}
\Big(Q I_v + \sdt W(v) A^B_{2,v}\Big) V^*
	&= \alpha^{-}V(\tau) - V^k, \\
\Big(P I_S - \sdt \rho_{s,v} W(S) A^B_{2,S} \Big)V^{k+1} &= V^*, \nonumber	\\
\alpha^{-} = (PQ + 1) I &- Q\sdt \rho_{s,v} W(S) A_{S}^F + P\sdt W(v) A_{v}^F. \nonumber
\end{align}
Then this scheme is unconditionally stable in time step $\Delta \tau$, approximates \eqref{spl2} with $O(\sqrt{\dtau}\max(h_S,h_v))$ and preserves positivity of the vector $V(x,\tau)$ if $Q = \beta \sdt/h_v, \ P = \beta \sdt/h_S$, where $h_v, h_S$ are the grid space steps correspondingly in $v$ and $S$ directions, and the coefficient $\beta$ must be chosen to obey the condition:
\[
\beta > \max_{S,v}[W(v) - \rho_{s,v} W(S)].
\]
\end{proposition}
\begin{proof}
The proof is completely analogous to that given for Proposition~\ref{propPos}, therefore we omit it for the sake of brevity.
\end{proof}

\subsubsection{Rate of convergence of the Picard iterations}

It would be interesting to estimate the rate of convergence of the proposed Picard iterations. For doing so, let us define the following operators
\begin{align} \label{defs}
T_1 &= Q + \sdt W(v) \triangledown_v, \qquad T_2 = P - \sdt \rho_{s,v} W(S) \triangledown_S, \\
T_3 &= PQ + 1 - Q\sdt \rho_{s,v} W(S)\triangledown_S + P\sdt W(v) \triangledown_v. \nonumber
\end{align}
The exact solution $V$ of \eqref{trickNeg} after the transformation described in \eqref{trickNeg1} is applied, can be re-written by using this notation in the form
\begin{equation} \label{exact}
T_1 T_2 V = T_3 V(\tau) + V,
\end{equation}
Also with this notation, the scheme described in Proposition~\ref{propPos} can be represented as
\begin{equation} \label{FDscheme}
T_1 T_2 V^{k+1} = T_3 V(\tau) + V^k.
\end{equation}
Subtracting \eqref{exact} from \eqref{FDscheme} we obtain
\begin{equation} \label{subtract}
V^{k+1} - V = T_2^{-1} T_1^{-1} (V^k-V) \equiv T (V^k-V).
\end{equation}
We can estimate the spectral norm of $T$
\begin{align} \label{norm}
\Vert T \Vert &= \Vert T_2^{-1} T_1^{-1} \vert \le \Vert T_2^{-1}\Vert  \Vert T_1^{-1} \Vert = \left| \dfrac{1}{\min \lambda_{1,i}}\right| \left|\dfrac{1}{ \min \lambda_{2,i}} \right|,
\end{align}
\noindent where $\Vert\cdot\Vert$ denotes the spectral norm of an operator, and $\lambda_{j,i}$ is a set of eigenvalues of the corresponding operator $T_i, \ j \in [1,2,3]$. Based on the Proposition~\ref{propPos}
\begin{align} \label{spNorm2}
\left| \dfrac{1}{\min \lambda_{1,i}}\right| \left|\dfrac{1}{ \min \lambda_{2,i}} \right|
&= \left|\dfrac{\dtau}{h_S h_v}\left(\beta + \min_v \frac{3}{2}W(v)\right)\left(\beta + \min_S \frac{3}{2}\rho_{s,v} W(S)\right)\right|^{-1} \le \dfrac{h_S h_v}{\dtau \beta^2}.
\end{align}
Thus, this map is contraction if $\beta^2 > h_s h_v /\dtau$, so $\Vert T \Vert < 1$.
Moreover, based on \eqref{gamma}, the coefficient $\beta$ can be made big enough by choosing a large value of $\gamma$, and, hence, $\Vert T \Vert$ becomes small. Therefore, the Picard iterations should converge pretty fast. Also the rate of the convergence is linear

\subsubsection{Second order of approximation in space}

The results formulated in the above two Propositions can be further improved by making the whole scheme to be of the second order of approximation in $h_S$ and $h_v$. We call this FD scheme as Scheme~B.

\begin{proposition} \label{propPos2}
Let us assume $\rho_{s,v} \ge 0$, and approximate the lhs of \eqref{spl2} using the following finite-difference scheme:
\begin{align} \label{fd13}
\Big(Q I_v + \sdt W(v) A^B_{2,v}\Big) V^*
	&= \alpha^+_2 V(\tau) - V^k, \\
\Big(P I_S - \sdt \rho_{s,v} W(S) A^F_{2,S} \Big)V^{k+1} &= V^*, \nonumber	\\
\alpha^{+}_2 = (PQ + 1)I &- Q\sdt \rho_{s,v} W(S) A_{2,S}^B + P\sdt W(v) A_{2,v}^F. \nonumber
\end{align}
Then this scheme is unconditionally stable in time step $\Delta \tau$, approximates \eqref{spl2} with $O(\max(h^2_S,h^2_v))$ and preserves positivity of the vector $V(x,\tau)$ if $Q = \beta \sdt/h_v, \ P = \beta \sdt/h_S$, where $h_v, h_S$ are the grid space steps correspondingly in $v$ and $S$ directions, and the coefficient $\beta$ must be chosen to obey the condition:
\[
\beta > \dfrac{3}{2}\max_{S,v}[W(v) + \rho_{s,v} W(S)].
\]
The scheme \eqref{fd13} has a linear complexity in each direction.
\end{proposition}
\begin{proof}
See Appendix~\ref{appPos2}.
\end{proof}

\begin{proposition} \label{propNeg2}
Let us assume $\rho_{s,v} \le 0$, and approximate the lhs of \eqref{spl2} using the following finite-difference scheme of the second order in space:
\begin{align} \label{fd23}
\Big(Q I + \sdt W(v) A^B_{2,v}\Big) V^*
	&= \alpha^-_2 V(\tau) - V^k, \\
\Big(P I_S - \sdt \rho_{s,v} W(S) A^B_{2,S} \Big)V^{k+1} &= V^*, \nonumber	\\
\alpha^{-}_2 = (PQ + 1) I &- Q\sdt \rho_{s,v} W(S) A_{2,S}^F + P\sdt W(v) A_{2,v}^F. \nonumber
\end{align}
Then this scheme is unconditionally stable in time step $\Delta \tau$, approximates \eqref{spl2} with $O(\sqrt{\dtau}\max(h_S,h_v))$ and preserves positivity of the vector $V(x,\tau)$ if $Q = \beta \sdt/h_v, \ P = \beta \sdt/h_S$, where $h_v, h_S$ are the grid space steps correspondingly in $v$ and $S$ directions, and the coefficient $\beta$ must be chosen to obey the condition:
\[
\beta > \dfrac{3}{4}\max_{S,v}[W(v) - \rho_{s,v} W(S)].
\]
The scheme \eqref{fd23} has a linear complexity in each direction.
\end{proposition}
\begin{proof}
The proof is completely analogous to that given for Proposition~\ref{propPos2}, therefore we omit it for the sake of brevity.
\end{proof}

Once again we want to underline that the described approach to deal with the mixed derivative term supplies just the first order approximation in time. But that is exactly what was done in the HV scheme as well. Nevertheless, the whole splitting scheme \eqref{HV} is of the second order in $\Delta \tau$.

The coefficient $\beta$ should be chosen experimentally. In our experiments described in the following sections we used
\begin{equation} \label{beta}
\beta = 10 \max_{S,v}[W(v) - \rho_{s,v} W(S)].
\end{equation}
For the second and third equations in \eqref{expSpl} similar Propositions can be used to solve these equations and guarantee the second order approximation in space, the first order approximation in time and positivity of the solution as well as the convergence of the Picard fixed point iterations. A small but important improvement, however, must be made for the second equation in \eqref{expSpl} since the definition of $F_{Sr}(\tau_{n-1})$ contains $\sqrt{v}$ which is a dummy variable for this equation. Accordingly, as this equation should be solved in a loop on $v_j, j=1,...,N_2$, where $v_j$ are the nodes on the $v$-grid, and $N_2$ is the number of these nodes, for each such a step its own $\beta_j$ must be computed based on the condition
\[ \beta_j > \max_{S,v}[W(v) - \rho_{s,v} W(S)\sqrt{v_j}]. \]
This, however, doesn't bring any problem.

Below, for the sake of convenience, we provide the explicit formulae of the first derivatives for the backward $D^{1}_{2B}$ and forward $D^{1}_{2F}$ approximations of the second order at a non-uniform-grid. They read \cite{Hout3D}
\begin{align*}
D^{1}_{2B} f(x)\Big|_{x = x_i} &=  f(x_i)\dfrac{h_i}{h_{i-1}(h_i + h_{i-1})}
- f(x_{i-1})\dfrac{h_{i-1}+h_i}{h_i h_{i-1}} + f(x_{i-2})\dfrac{h_{i-1} + 2 h_i}{h_{i}(h_{i-1} + h_i)},
\\
D^{1}_{2F} f(x)\Big|_{x = x_i} &= - f(x_i)\dfrac{2 h_{i+1} + h_{i+2}}{h_{i+1} h_{i+2}}
+ f(x_{i+1})\dfrac{h_{i+2}+h_{i+1}}{h_{i+2}h_{i+1}} - f(x_{i+2})\dfrac{h_{i+1}}{h_{i+2}(h_{i+2} + h_{i+1})},
\end{align*}
\noindent where $h_i = x_i - x_{i-1}$. Based on this definition, the matrices $A_2^B, A_2^F$ can be constructed accordingly.

\subsubsection{Fully implicit scheme}

For even better stability, the whole first step \eqref{twoStep} of the HV scheme can be made fully implicit. In doing that observe, that the first line in \eqref{HV} is a \pade approximation (0,1) of the equation
\begin{equation} \label{pdePade}
\fp{V(\tau)}{\tau} = [F_0(\tau) + F_1(\tau) + F_2(\tau) + F_3(\tau)]V(\tau).
\end{equation}
The solution of this equation can be obtained as
\begin{align} \label{twoStepI}
V(\tau) &= \exp\left\{\Delta \tau \left[F_0(\tau_{n-1}) + F_1(\tau_{n-1}) +
F_2(\tau_{n-1}) + F_3(\tau_{n-1})\right]\right\}V(\tau_{n-1}) \\
&= e^{\Delta \tau F_0(\tau_{n-1})}e^{\Delta \tau F_1(\tau_{n-1})}
e^{\Delta \tau F_2(\tau_{n-1})}e^{\Delta \tau F_3(\tau_{n-1})}V(\tau_{n-1})
+ O(\Delta \tau). \nonumber
\end{align}

Alternatively, a \pade approximation (1,0) can also be applied to all exponentials in \eqref{twoStepI} providing same order of approximation in $\Delta \tau$ but making all steps implicit. Namely, this  results to the following splitting scheme of the solution of \eqref{pdePade}:
\begin{align} \label{impl1}
[1 - \Delta \tau F_0(\tau)]V^0 &= V(\tau_{n-1}), \\
[1 - \Delta \tau F_1(\tau)]V^1 &= V^0(\tau_{n-1}), \nonumber \\
[1 - \Delta \tau F_2(\tau)]V^1 &= V^1(\tau_{n-1}), \nonumber \\
[1 - \Delta \tau F_3(\tau)]V(\tau) &= V^2(\tau_{n-1}). \nonumber
\end{align}
We already know how to solve the first step in \eqref{impl1} (which always was a bottleneck for applying this fully implicit scheme). The remaining steps (lines 2-4 in \eqref{impl1}) can be done similar to the steps 2-4 (line 2 in \eqref{HV}) in the HV scheme. Thus, the whole first step in the HV scheme becomes implicit while has the same linear complexity in the number of nodes. Also our experiments confirm that this scheme provides great stability and preserves positivity of the solution. Therefore, running first few Rannacher steps is not necessary.

The third line in \eqref{HV} can be modified accordingly as follows:
\begin{align} \label{IntStep}
\tilde{Y}_0 &= Y_0 + \dfrac{1}{2} \Delta \tau \left[F(\tau_n) Y_k - F(\tau_{n-1})V_{n-1}\right], \\
&= Y_0 + \dfrac{1}{2} \left[Y_3 + \Delta \tau F(\tau_n Y_3\right] - \dfrac{1}{2} \left[V_{n-1} + \Delta \tau F(\tau_{n-1}) V_{n-1}\right] - \dfrac{1}{2}Y_3 + \dfrac{1}{2}V_{n-1} \nonumber \\
&= Y_0 + \dfrac{1}{2}\left[\tilde{Y}_3 - Y_0 - Y_3 + V_{n-1}\right]. \nonumber
\end{align}
Here all values in the rhs of this equation are already known except of $\tilde{Y}_3$ which is the solution of the problem $\fp{Y_3}{\tau} = F(\tau_n) Y_3$. Therefore, it can be solved in the exactly same way as the first step of our fully implicit scheme.

Note that both Scheme~A and Scheme~B can be used as a part of the fully implicit scheme. If one makes a choice in favor of Scheme~A, then the situation is as follows. We have the first step of the fully implicit scheme done with approximation $O(\sqrt{\dtau} \min(h_S,h_v))$, and then multiple steps (six  for the 3D problem) with the second order in space, so the total approximation is expected to be close to two. For the second sweep of the HV scheme this is same. If the choice is made in favor of Scheme~B, the rigorous spatial approximation of the whole HV scheme becomes two.

An obvious disadvantage of the proposed schemes is some degradation of performance, since it requires some number of iterations to converge\footnote{In our experiments 1-2 iterations were sufficient to provide the relative tolerance to be $10^{-6}$.}
when computing the mixed derivatives step, and at every iteration we need to solve 2 systems of linear equations. Despite the total complexity is still linear in the number of nodes, it takes about 4 times more computational time than the explicit scheme. However, as we have already mentioned, in our experiments the explicit scheme at the first and the third steps of \eqref{HV} might suffer from instability (which becomes more pronounced  when the dimensionality of the problem increases), i.e., either it requires a very small and impractical temporal step to converge, or, otherwise, it explodes. Also our results show that the proposed scheme is only about 50-70\% slower than the explicit step of the original HV scheme. However, the time step of our scheme can be significantly increased while still
maintaining the stability of the scheme, whereas this could be problematic for the HV scheme. Therefore, this increase in the time step could compensate the extra time required for doing the first step implicitly. For instance, running one step in time for the 3D advection-diffusion problem using the HV scheme coded in Matlab, at our machine takes 2 secs, while the fully implicit scheme requires 2.6 secs. On contrary, the HV scheme behaves kind of unstable with no Rannacher steps even with $\Delta \tau = 0.005$ yrs, while the fully implicit scheme continues to work well, e.g., with $\Delta \tau = 0.05$ yrs\footnote{Using the same $\beta$ as in the above. However, changing the first multiplier in the rhs of \eqref{beta} can make the scheme working for the higher values of the time step as well.}. So if by the accuracy reason this step is sufficient, it can improve performance by factor 10, and then loosing about 50-70\% for the implicit scheme is not sensitive.

Note, that the same approach can be applied to the forward equations based on the approach of \cite{Itkin2014b}. In more detail this will be presented elsewhere.

Once this step is accomplished, the whole scheme \eqref{HV} becomes positivity-preserving. This is because for the first step of \eqref{HV} our fully implicit scheme does preserve positivity. The next steps can be re-written in the form
\[ [1 - \theta \Delta \ F_j(\tau_n)] Y_j = Y_{j-1} - \theta \Delta \tau F_j(\tau_{n-1}) V_{n-1} \]
So $\Delta \tau$ can always be chosen small enough to make this step positivity-preserving, if the lhs matrix  $M = 1 - \theta \Delta \ F_j(\tau_n)$ is an EM-matrix. The latter can be achieved by taking an approach of \cite{Itkin2014}, where the first spatial derivatives are discretized by using one-sided finite-differences with the second order of approximation. Same is true for the second sweep of the splitting steps in \eqref{HV}.

\subsection{Jump steps}
Going back to the general splitting scheme \eqref{splitting}, observe that the one-dimensional jumps are treated at steps 2-4 and 6-8. Obtaining solutions at these steps for some popular L{\'e}vy models such as Merton, Kou, CGMY (or GTSP), NIG, General Hyperbolic and Meixner ones, could be done as it is shown in \cite{Itkin2014, Itkin2014a}. Efficiency of this method in general is not worse than that of the Fast Fourier Transform (FFT), and in many cases is linear in $N$ - the number of the grid
points. In particular, this is the case for the Merton, Kou, CGMY/GTSP at $\alpha \le 0$ and Meixner models.

In this paper we model all ideosyncratic and common jumps by using the Meixner model. So below we sequentially consider all jump steps of the splitting algorithms in \eqref{splitting}.

\subsubsection{Idiosyncratic jumps}
Remember, that the characteristic exponent of the Meixner process is
\begin{equation} \label{MCE}
\phi(u, a,b,d,m) = 2d \left\{ \log [\cos(b/2)] - \log \left[ \cosh\left( \dfrac{a u - i b}{2} \right) \right] \right\} + i m u,
\end{equation}
\noindent and the L{\'e}vy density $\mu(dy)$ of the Meixner process reads
\begin{equation*}
\mu(dy) = d \dfrac{\exp(b y/a)}{y \sinh(\pi y/a)}d y.
\end{equation*}
Therefore, from \eqref{Jump} we immediately obtain
\begin{equation} \label{MJ}
\mathcal{J} = \phi(-i \triangledown, a,b,d,m) =  2d \left\{ \log [\cos(b/2)] - \log \left[ \cos\left( \dfrac{a \triangledown + b}{2} \right) \right] \right\} + m \triangledown.
\end{equation}

The discretization scheme for this operator which provides second order of approximation in space and time while preserves positivity of the solution is given in \cite{Itkin2014a}, Propositions 3.8, 3.9.

At the end of this section we also remind, that according to the method of \cite{Itkin2014a} the drift term in  \eqref{MJ} (the last one) could be either moved into the drift term of the corresponding advection-diffusion operator, or could be discretized as
\begin{equation} \label{cas2}
e^{\dtau m \triangledown_\chi} =
\begin{cases}
e^{\dtau m A^F_{2,\chi}} + O(h^2_\chi), & m > 0, \\
e^{\dtau m A^B_{2,\chi}} + O(h^2_\chi), & m < 0.
\end{cases}
\end{equation}
This is possible because in both cases in \eqref{cas2} the exponent is the negative of the EM-matrix\footnote{Here EM is an abbreviation for an eventually M-matrix, see \cite{Itkin2014a}.}, therefore $e^{\dtau m \triangledown_\chi}$ is a positive matrix with all eigenvalues $|\lambda_i| < 1$.

\subsubsection{Common jumps}

The most difficult step is to solve the problem
\begin{equation} \label{J123}
V(\tau + \Delta \tau) = e^{\Delta \tau \mathcal{J}_{123}}V(\tau).
\end{equation}
In \cite{ItkinLipton2014} it was demonstrated how to do this when the common jumps are represented by the Kou model using a modification of the Peaceman-Rachford ADI method, see \cite{McDonough2008}. Here we shortly describe the algorithm for the Meixner model.

Remember that by definition $\mathcal{J}_{123}$ is given by \eqref{MJ} where now $\triangledown = b_s \triangledown_s + b_v \triangledown_v + b_r \triangledown_r$. The drift term $m \triangledown$  again can be split among the corresponding drifts of the diffusion operators. After that we need to solve the following equation (\cite{Itkin2014a})
\begin{align} \label{ADIeq}
\prod_{n=1}^\infty M_n^{\kappa}V(\tau &+ \dtau) = [\cos(b/2)]^{\kappa} V(\tau), \\
M_n &= 1 - \dfrac{(a \triangledown + b)^2}{4 \pi^2 (n- 1/2)^2}, \qquad \kappa = 2 d \dtau, \nonumber
\end{align}
\noindent where the parameters $a, b, d$ characterize the common jumps.

This equation can be solved in a loop on $n$. Namely, we start with $n=1$ and take $V_0 = [\cos(b/2)]^\kappa V(\tau)$.
Since in our experiments $0 < \kappa < 1$,\footnote{This can always be achieved by choosing a relatively small $\dtau$.} at every step in $n$ we run this scheme for $\kappa = 0,1$ and then use linear interpolation to $\kappa$. At $\kappa = 0$ an obvious solution is $V(\tau + \dtau) = V(\tau)$. At $\kappa = 1$ \eqref{ADIeq} is a 3D parabolic equation that can be solved using our implicit version of the HV scheme. Indeed, it can be re-written in the form
\[ \left[1 - (\dtau)^2 K_n \left(\triangledown +\frac{b}{a}\right)^2\right]V = V(\tau), \qquad K_n = \dfrac{a^2}{4 \pi^2(n-1/2)^2 (\dtau)^2}. \]
As usually $a$ is small, e.g., in \cite{Schoutens01} $a = 0.04$, so even for $n=1$ $K_n = O(1)$. Now using the \pade approximation theory we can re-write this equation as
\[ V = e^{(\dtau)^2 K_n \left(\triangledown +\frac{b}{a}\right)^2} V(\tau) + O((\dtau)^2). \]
\noindent Therefore, if we omit the last term $O((\dtau)^2)$, the total second order approximation of the scheme in time is preserved. This latter equation is equivalent to
\begin{equation} \label{JumpDif}
\fp{V}{s} = \left(\triangledown +\frac{b}{a}\right)^2 V,
\qquad V(0) = \cos(b/2) V(\tau), \quad s \in [0, T_n],
\end{equation}
\noindent which has to be solved at the time horizon (maturity) $T_n = (\dtau)^2 K_n = \frac{a^2}{4 \pi^2(n-1/2)^2}$. Since $T_1$ is small and usually less than $\dtau$ we may solve it in one step in time. And when $n$ increases, this conclusion remains to be true as well.

Once this solution is obtained we proceed to the next $n$. Thus, this scheme runs in a loop starting with $n=1$ and ending at some $n=M$. Similar to how we did it for the idiosyncratic jumps we choose $M=10$ based on the argument of \cite{Itkin2014a}, namely: i) the high order derivatives of the option price drop down pretty fast in value\footnote{For instance, for the plain vanilla options the option price asymptotically is limited by the intrinsic value, therefore high order derivatives rapidly vanish.}, and ii) first 10 terms of the sum  $\sum_{i=1}^\infty T_i$ approximate the whole sum with the accuracy of 1\%. The solution obtained after $M$ steps is the final solution.

Overall, the whole splitting algorithm contains 11 steps. The complexity of each step is linear in $N$ since at every step we solve some parabolic equation with a tridiagonal or pentadiagonal matrix. Thus, the total complexity of the method is $\varsigma N_1 N_2 N_3$ where $N_i$ is the number of grid nodes in the $i$-th dimension, and $\varsigma$ is some constant coefficient, which is about 276 (18 systems for one diffusion step if the implicit modification of the HV scheme is used times 2 diffusion steps, so totally 36; 10 systems per a 1D jump step times 2 steps times 3 variables, so totally 60; 18 steps per a single 3D parabolic PDE solution for common jumps times 10 steps, so totally 180).

Still this could be better than a straightforward application of the FFT (in case the FFT is
applicable, e.g., the whole characteristic function is known in closed form which is not the case if one takes into account local volatility, etc.) which usually requires the number of FFT nodes to be a power of 2 with a typical value of 2$^{11}$. It is also better than the traditional approach which considers approximation of the linear non-local jump integral $\mathcal{J}$ on some grid and then makes use of the FFT to compute a matrix-by-vector product. Indeed, when using FFT for this purpose we need two sweeps per dimension using a slightly extended grid (with, say, the tension coefficient $\xi$) to avoid wrap-around effects, \cite{Halluin2005b}. Therefore the total complexity per time step could be at least $O( 8 \xi_1 \xi_2 \xi_3 N_1 N_2 N_3 \log_2 (\xi_1 \xi_2 \xi_3 N_1 N_2 N_3))$ which for the FFT grid with $N_1 = N_2 = N_3 = 2048$, and $\xi_1 = \xi_2 = 1.3$ is 2.5 times slower than our method\footnote{As mentioned by the referee, since the flop counts rarely predict accurately an elapsed time, this statement should be further verified.}. Also the traditional approach experiences some other problems for jumps with infinite activity and infinite variation, see survey in \cite{Itkin2014} and references therein. Also as we have already mentioned using Fast Gauss Transform for the common jump step could significantly reduce the time for this most time-consuming piece of the splitting scheme.

\section{Numerical experiments\label{numExamp}}

Due to the splitting nature of our entire algorithm represented by \eqref{splitting}, each step of splitting is computed using a separate numerical scheme. All schemes provide second order approximation in both space and time, are unconditionally stable and preserve positivity of the solution.

In our numerical experiments for the steps which include mixed derivatives terms we used the suggested fully implicit version of the Hundsdorfer-Verwer scheme. This allows one to eliminate any additional damping scheme of the lower order of approximation, e.g., implicit Euler scheme (as this is done in the Rannacher method), or Do scheme with the parameter $\theta=1$ (as this was suggested in \cite{Hout3D}).

A non-uniform finite-difference grid is constructed similar to \cite{HoutFoulon2010} in $v$ and $r$ domains, and as described in \cite{ItkinCarrBarrierR3} in the $S$ domain. In case of barrier options we extended the $S$ grid by adding 2-3 ghost points either above the upper barrier or below the lower barrier, or both with the same boundary conditions as at the barrier (rebate or nothing). Construction of the jump grid, which is a superset of the finite-difference grid used at the first (diffusion) step is also described in detail in \cite{Itkin2014}. Normally the diffusion grid contained 61 nodes in each space direction. The extended jump grid contained extra 20-30 nodes. If a typical spot value at time $t=0$ is $S_0$=100, the full grid started at $S=0$ and ended up at $S=10^3$.

We computed our results in Matlab at a standard PC with Intel Xeon E5620 2.4 Ghz CPU. A typical elapsed time for computing one time step for the pure diffusion model with no jumps is given in the Table~\ref{elapsed}\footnote{Note, that, e.g., for the HV scheme we need 2 sweeps per one step in time.}:
\begin{table}[!htb]
\begin{minipage}{\linewidth}
\begin{center}
\begin{tabular}{|c|c|c|c|c|c|}
\hline
$N$ of nodes & \multicolumn{4}{|c|}{Advection-Diffusion}  \\
\cline{2-5}
             & Mixed der & 1D steps & Total for 1 sweep & k   \\
\hline
50x50x50 & 0.81 & 0.38 &  1.19 &  - \\
\hline
60x60x60 & 1.26 & 0.59 &  1.85 & 2.42   \\
\hline
70x70x70 & 1.88 & 0.86 &  2.74 & 2.54  \\
\hline
80x80x80 & 2.71 & 1.28 &  3.89 & 2.62 \\
\hline
100x100x100 & 4.50 & 2.22 &  3.89 & 3.17 \\
\hline
\end{tabular}
\caption{Elapsed time in secs for 1 step in time to compute the advection-diffusion problem.}
\label{elapsed}
\end{center}
\end{minipage}
\end{table}
Here $k = \log[t_{i}/t_{i-1}]/\log[N_i/N_{i-1}]$ is the power in the complexity $C$ of calculations, which is regressed to $C \propto N^k$. It can be seen that the complexity is almost linear in all dimensions regardless of the number of nodes. The slight growth of $k$ can be attributed to the way how Matlab processes large sparse matrices. Also our C++ implementation is about 15 times faster than Matlab.

\paragraph{European call option} In this test we solved an European call option pricing problem using the described model in a pure diffusion context, hence all jump intensities are set to zero. Also for simplicity we assumed all parameters of the model to be time-independent. Thus, in this test the robustness of our convection-diffusion FD scheme is validated\footnote{In this paper we don't analyze the convergence and order of approximation of the FD scheme, since the convergence in time is same as in the original HV scheme, and approximation was proven by the Theorem. For the jump FD schemes the convergence and approximation are considered in \cite{Itkin2014}.}. Parameters of the model used in this test are given in Table~\ref{Tab13D}\footnote{The values of parameters are taking just for testing, and could differ from those, e.g., obtained by calibrating the model to the current market data.}, and the results are presented in Fig.~\ref{Fig3DXY},\ref{Fig3DXZ},\ref{Fig3DYZ}. We choose $a = b = 0.5, \ c = 1$, and the local volatility function was set to 1, so pretty much in this test our model is a lognormal + double CIR model (with stochastic volatility and stochastic interest rate). The computational domain for the diffusion part is $S \in [0,1000], \ v \in [0,3], \ r \in [0,1]$, and for the jump part - $S \in [0,10000], \ v \in [0,50], \ r \in [0,10]$.

\begin{table}[!htb]
\begin{center}
\scalebox{0.9}{
\begin{tabular}{|c|c|c|c|c|c|c|c|c|c|c|c|c|c|c|c|}
\hline
$T$ & $K$ & $\kappa_V$ & $\xi_v$ & $\theta_v$ & $\kappa_r$ & $\xi_r$ & $\theta_r$ & $q$ & $\rho_{Sv}$
& $\rho_{Sr}$ & $\rho_{vr}$ & $\phi_{Sv}$ & $a$ & $b$ \\
\hline
1 & 100 & 2 & 0.3 & 0.9 & 3 & 0.1 & 0.05 & 0.5 & -0.647 & 0 & 0.1 & $4\pi/5$ & 0.5 & 0.5 \\
\hline
\end{tabular}
}
\caption{Parameters of the test for pricing an European call option.}
\label{Tab13D}
\end{center}
\end{table}

We recall that a correlation matrix $\Sigma$ of $N$ assets can be represented as a Gram matrix
with matrix elements $\Sigma_{ij} = \left<{\bf x}_i,{\bf x}_j\right>$ where ${\bf x}_i, {\bf x}_j$ are unit vectors on a $N-1$ dimensional hyper-sphere $\mathbb{S}^{N-1}$. Using the 3D geometry, it is easy to establish the following cosine law for the correlations between three assets:
\[ \rho_{xy} = \rho_{yz} \rho_{xz} + \sqrt{(1-\rho_{yz}^2)(1-\rho_{xz}^2)}cos(\phi_{xy}), \]
\noindent with $\phi_{xy}$ being an angle between $\bf x$ and its projection on the plane spanned by ${\bf y}, {\bf z}$. As discussed, e.g., by \cite{Dash2004}, three variables $\rho_{yz}, \rho_{zz}, \phi_{xy}$ are independent, but $\rho_{xy}, \rho_{xz}, \rho_{yz}$ are not. Therefore, the value $\rho_{Sv}$ in Table~\ref{Tab13D} was computed using given $\rho_{Sr}, \ \rho_{vr}$ and $\phi_{Sv}$.

Since the whole picture in this case is 4D, we represent it as a series of 3D projections, namely: Fig.~\ref{Fig3DXY} represents the $S-v$ plane at various values of the $r$ coordinate which are indicated in the corresponding labels; Fig.~\ref{Fig3DXZ} does same in the $S-r$ plane, and Fig.~\ref{Fig3DYZ} - in the $v-r$ plane. Based on the values of parameters of the FD scheme the expected discretization error should be about 1\%.

\begin{figure}[!htb]
\begin{center}
\includegraphics[width=0.8\textwidth, height=3in]{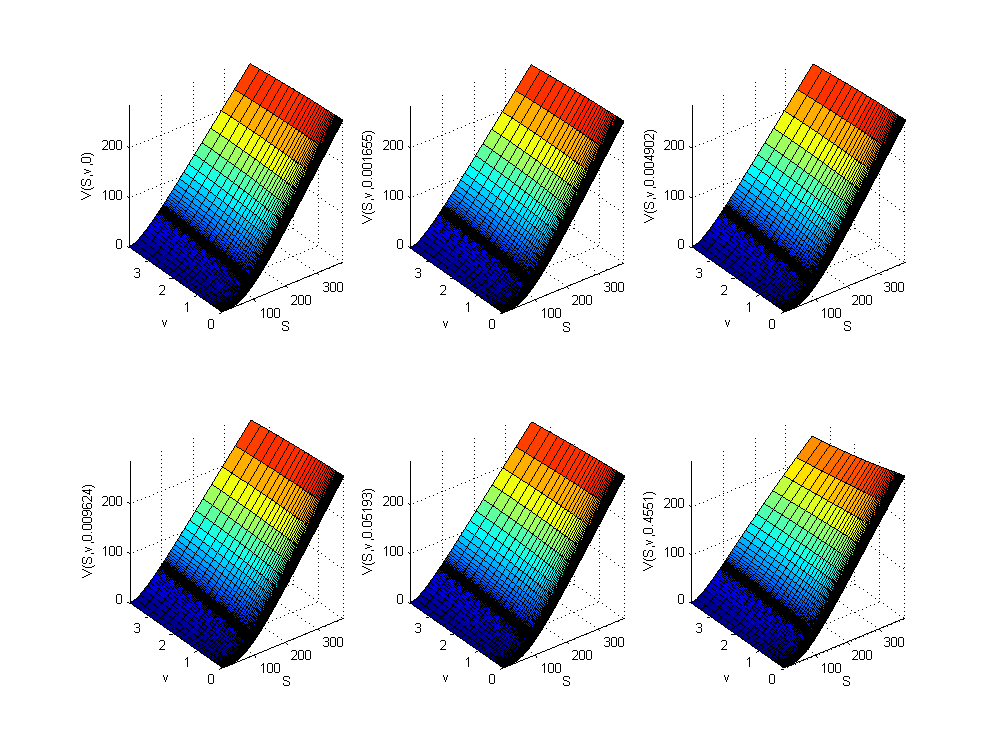}
\caption{European call option prices in $S_0-v_0$ plane at various values of $r_0$.}
\label{Fig3DXY}
\end{center}
\end{figure}

\begin{figure}[!htb]
\begin{center}
\includegraphics[width=0.8\textwidth, height=3in]{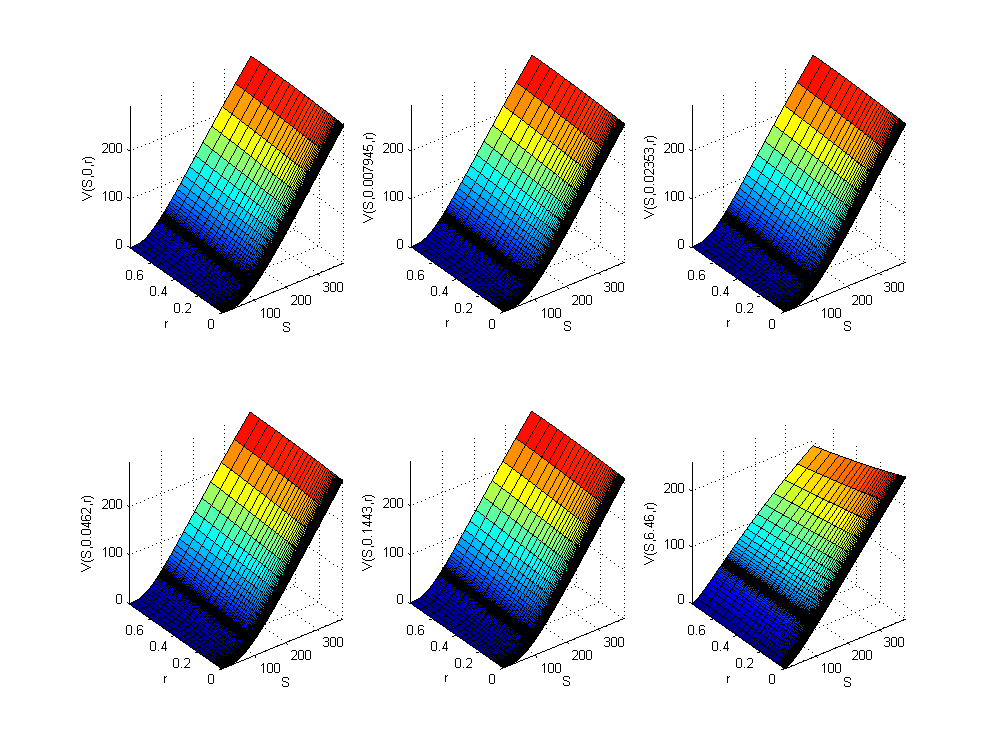}
\caption{European call option prices in $S_0-r_0$ plane at various values of $v_0$.}
\label{Fig3DXZ}
\end{center}
\end{figure}

\begin{figure}[!htb]
\begin{center}
\includegraphics[width=0.8\textwidth, height=3in]{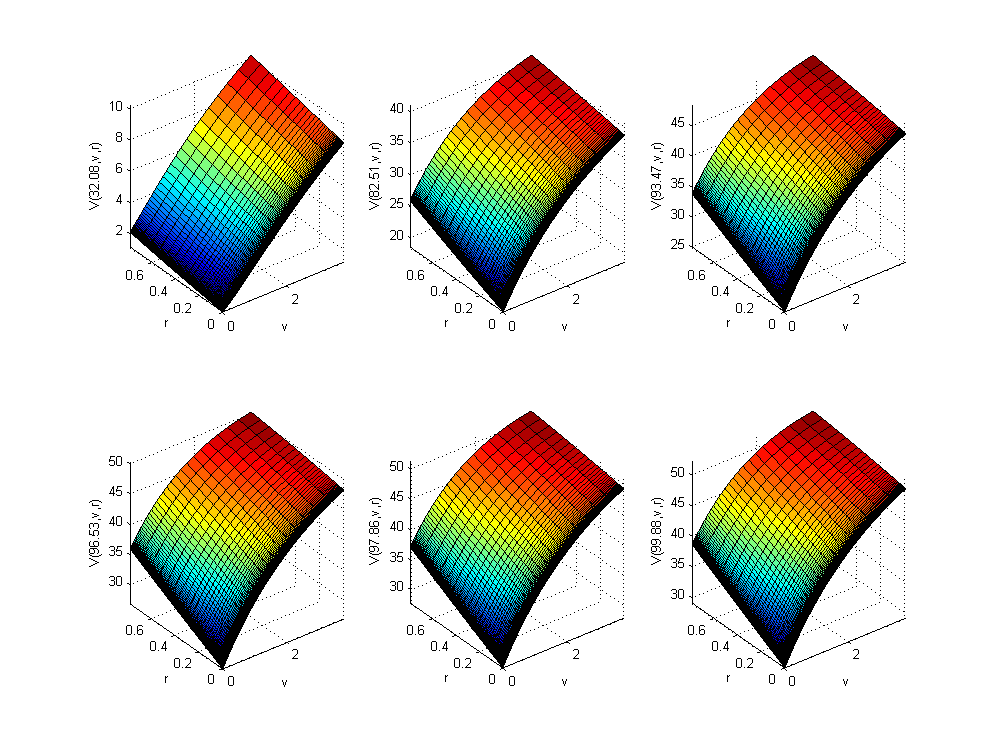}
\caption{European call option prices in $v_0-r_0$ plane at various values of $S_0$.}
\label{Fig3DYZ}
\end{center}
\end{figure}

\paragraph{Double barrier option} In this test we solved a more challenging problem of pricing double barrier option using the same model with no jumps with the lower barrier $L=50$ and the upper barrier $H=130$. Parameters of the model used in this test are given in Table~\ref{Tab23D}, and the results are presented in Fig.~\ref{Fig3BXY},\ref{Fig3BXZ},\ref{Fig3BYZ}.
\begin{table}[!htb]
\begin{center}
\scalebox{0.9}{
\begin{tabular}{|c|c|c|c|c|c|c|c|c|c|c|c|c|c|c|c|}
\hline
$T$ & $K$ & $\kappa_V$ & $\xi_v$ & $\theta_v$ & $\kappa_r$ & $\xi_r$ & $\theta_r$ & $q$ & $\rho_{Sv}$
& $\rho_{Sr}$ & $\rho_{vr}$ & $\phi_{Sv}$ & $\alpha$ & $\beta$ \\
\hline
0.5 & 100 & 2.5 & 0.3 & 0.6 & 0.3 & 0.1 & 0.05 & 0 & -0.587 & 0.3 & 0.4 & $4\pi/5$ & 0.5 & 0.5 \\
\hline
\end{tabular}
}
\caption{Parameters of the test for pricing a Double barrier call option.}
\label{Tab23D}
\end{center}
\end{table}

\begin{figure}[!htb]
\begin{center}
\includegraphics[width=0.8\textwidth, height=3in]{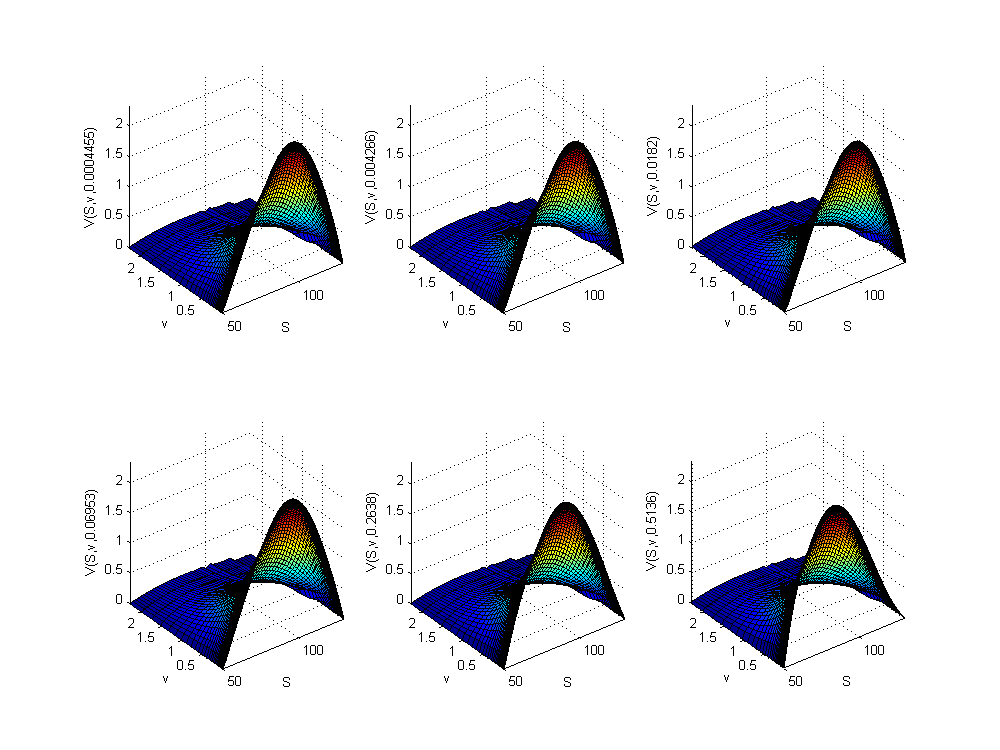}
\caption{Double barrier call option prices in $S_0-v_0$ plane at various values of $r_0$.}
\label{Fig3BXY}
\end{center}
\end{figure}

\begin{figure}[!htb]
\begin{center}
\includegraphics[width=0.8\textwidth, height=3in]{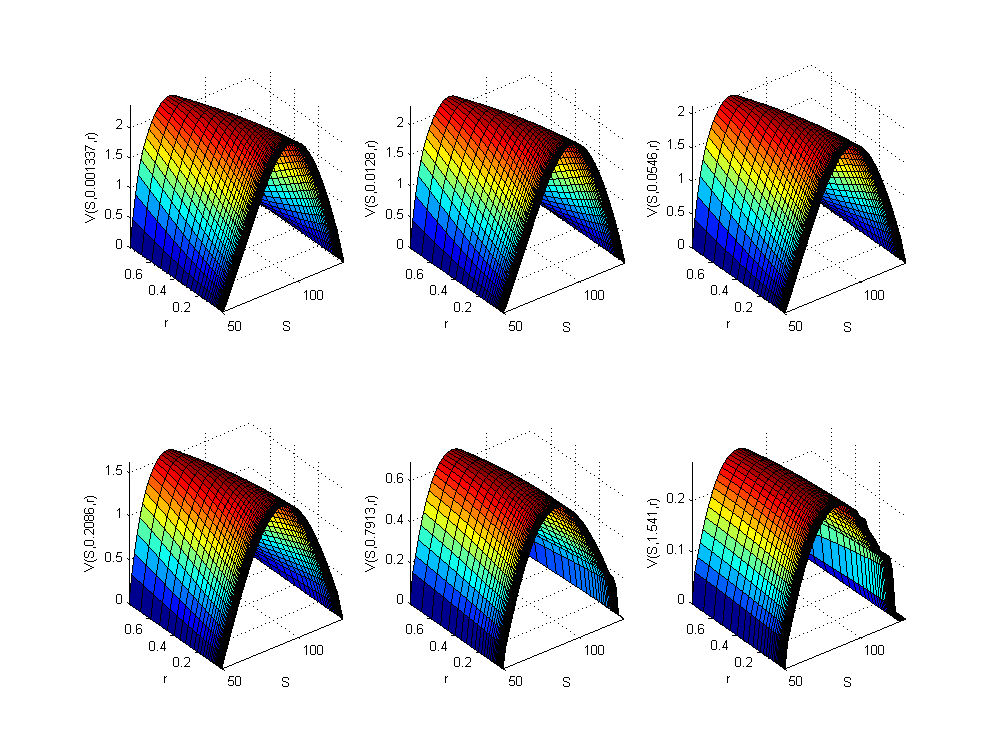}
\caption{Double barrier call option prices in $S_0-r_0$ plane at various values of $v_0$.}
\label{Fig3BXZ}
\end{center}
\end{figure}

\begin{figure}[!htb]
\begin{center}
\includegraphics[width=0.8\textwidth, height=3in]{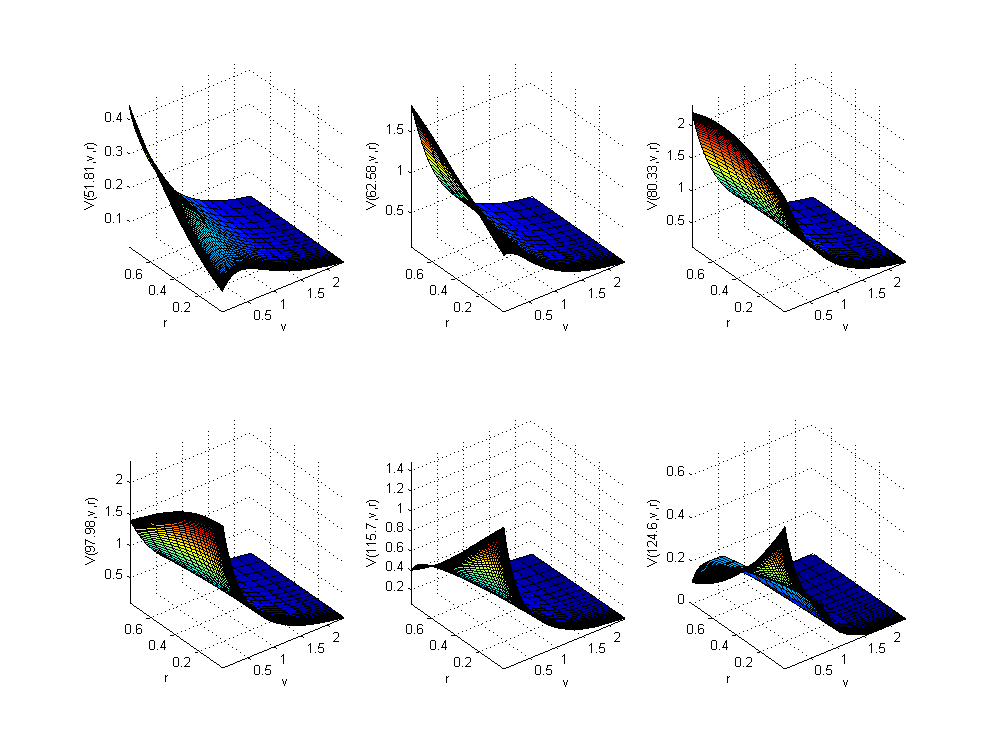}
\caption{Double barrier call option prices in $v_0-r_0$ plane at various values of $S_0$.}
\label{Fig3BYZ}
\end{center}
\end{figure}

It can be observed that the damping properties of the fully implicit HV scheme seem sufficient. No spurious oscillations are observed, and the solution is monotone even near the critical points (close to strike and both barriers in $S$ space). However, some minor oscillations occur at high $v$ close to the upper barrier.

\paragraph{Up-and-Out call option with jumps} The third test deals with jumps using the Meixner model for both idiosyncratic and common jumps as was described in the previous section. Parameters of the jump processes are given in Table~\ref{jumpPar}, while the remaining parameters are the same as in Table~\ref{Tab23D}. The loading factors (as they are defined in Proposition~\ref{propBallotta}) we used in the test are: $b_S = 1, \ b_v = 2, \ b_r = 3$.
\begin{table}[!ht]
\begin{center}
\begin{tabular}{|c|c|c|c|c|}
\hline
Driver & $a$ & $b$ & $m$ & $d$  \\
\hline
$S$ & 0.04 & -0.33 & 0.1 & 52  \\
\hline
$v$ & 0.02 & -0.5 & 0.03 & 40 \\
\hline
$r$ & 0.01 & -0.2 & 0.01 & 30 \\
\hline
Common jumps & 0.03 & -0.1 & 0.05 & 40 \\
\hline
\end{tabular}
\caption{Parameters of the jump models.}
\label{jumpPar}
\end{center}
\end{table}

A typical elapsed time for computing one time step for the pure jump model is given in Table~\ref{elapsedJ}. Here we define the power $k$ assuming that the complexity $C$ is proportional to $(N_1 N_2 N_3)^\kappa$, so $\kappa$ can be found as
\[ \kappa = \log\left( \frac{t_i}{t_{i-1}}\right)\Big/
\log\left(\frac{N_{1i} N_{2i} N_{3i}}{N_{1,i-1} N_{2,i-1} N_{3,i-1}}\right). \]
One can see that in all experiments $\kappa$ is close to 1, so the complexity is linear in the number of nodes.

The results computed in this experiment are presented in Fig.~\ref{Fig3JXY},\ref{Fig3JXZ},\ref{Fig3JYZ} as a difference between the full case with the correlated jumps and diffusion and that with no jumps. It is clear that jumps can play a significant role changing the whole 4D profile of the option price. Varying the loading factors one can change the correlations between jumps, and thus affect the price in a significant degree. For instance, increasing all the loading factors in this experiment by factor 10 results to the decrease of the Up-and-Out option price almost to few cents. Thus, the proposed model is very flexible. However, calibration of all the model parameters can be very time-consuming. Therefore, it is better to calibrate various pieces of the model separately, as this was discussed, e.g., in \cite{Ballotta2014}. Namely, the idiosyncratic jumps first can be calibrated separately to some marginal distributions using the appropriate instruments. Then the parameters of the common  jumps can be calibrated to the option prices, while keeping parameters of the idiosyncratic jumps fixed.

\begin{table}[!htb]
\begin{center}
\begin{tabular}{|l|r|r|r|c|}
\hline
& \multicolumn{4}{|c|}{Jumps}  \\
\cline{2-5}
$N$ of nodes & Common & all 1D & $T_{1s}$ & $\kappa$   \\
\hline
114x95x84 & 70.6 & 3.26 & 77.1  &  - \\
128x95x84 & 80.1 & 4.72 & 89.5  &  1.29 \\
142x95x84 & 84.9 & 5.22 & 95.3  &  0.60 \\
156x95x84 & 91.7 & 5.83 & 103.4  & 0.87  \\
\hline
114x95x84 & 70.6 & 3.26 & 77.1  &  - \\
114x109x84 & 80.3 & 4.63 & 89.9  & 1.12  \\
114x123x84 & 91.9 & 5.30 & 102.5 & 1.09   \\
114x136x84 & 101.6 & 5.92 & 113.4 & 1.01  \\
\hline
114x95x84 & 70.6 & 3.26 & 77.1  &  - \\
114x95x98 & 79.7 & 4.69 & 89.1  & 0.94 \\
114x95x111 & 89.3 & 5.27 & 99.8 & 0.91 \\
114x95x123 & 98.2 & 5.88 & 110.0 & 0.95 \\
\hline
\end{tabular}
\caption{Elapsed time in secs for 1 full time step in time $T_{1s}$ to compute the 3D jump problem.}
\label{elapsedJ}
\end{center}
\end{table}

\begin{figure}[!htb]
\begin{center}
\includegraphics[width=0.8\textwidth, height=3in]{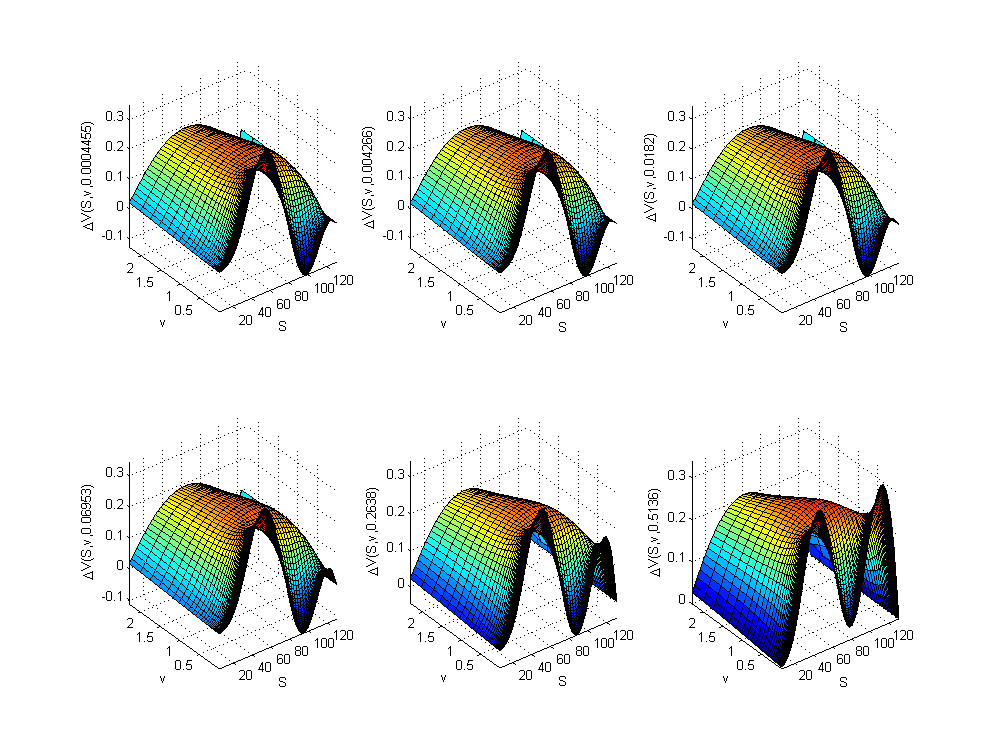}
\caption{Difference in Up-and-Out call option prices computed with and without jumps in $S_0-v_0$ plane at various values of $r_0$.}
\label{Fig3JXY}
\end{center}
\end{figure}

\begin{figure}[!htb]
\begin{center}
\includegraphics[width=0.8\textwidth, height=3in]{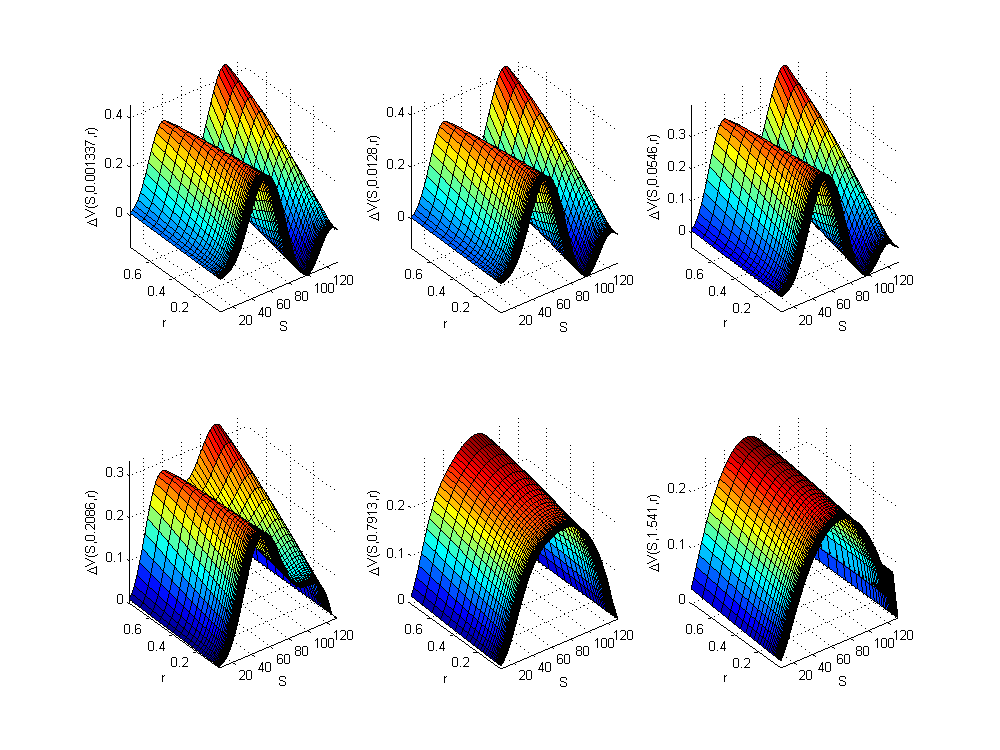}
\caption{Difference in Up-and-Out call option prices computed with and without jumps in $S_0-r_0$ plane at various values of $v_0$.}
\label{Fig3JXZ}
\end{center}
\end{figure}

\begin{figure}[!htb]
\begin{center}
\includegraphics[width=0.8\textwidth, height=3in]{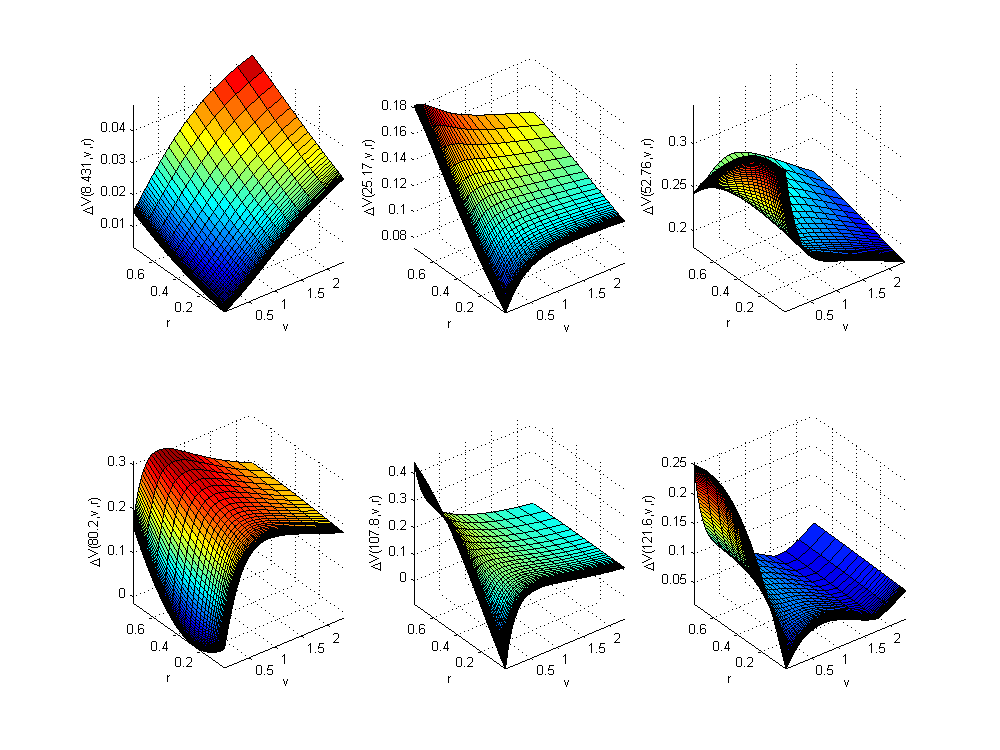}
\caption{Difference in Up-and-Out call option prices computed with and without jumps in $v_0-r_0$ plane at various values of $S_0$.}
\label{Fig3JYZ}
\end{center}
\end{figure}

\clearpage
\section{Conclusion}
In this paper we apply the approach of \cite{ItkinLipton2014} for pricing credit derivatives to various option pricing problems (vanilla and exotic) where as an underlying model we use Local stochastic volatility model with stochastic interest rates and jumps in every stochastic driver. It is important that all jumps as well as the Brownian motions are correlated. Here we solve just the backward problem (solving the backward Kolmogorov equation, e.g., for pricing derivatives), while the forward problem (solving the forward Kolmogorov equation to find the density of the underlying process) can be treated in a similar way, see \cite{Itkin2014b}.

In \cite{ItkinLipton2014} test examples were given for the Kou and Merton models, while the approach is in no way limited by these models. Therefore, in this paper we demonstrate how a similar approach can be used together with the Meixner model. Again, this model is chosen only as an example, because, in general, the approach in use is rather universal. We provide an algorithm and results of numerical experiments.

The second contribution of the paper is a new fully implicit modification of the popular Hundsdorfer and Verwer and Modified Craig-Sneyd finite-difference schemes which provides second order approximation in space and time, is unconditionally stable and preserves positivity of the solution, while still keeps a linear complexity in the number of grid nodes. This scheme has extended damping properties, and, therefore,
allows an elimination of any additional damping scheme of a lower order of approximation, e.g., implicit Euler scheme (as this is done in the Rannacher method), or Do scheme with the parameter $\theta=1$ (as this was proposed in \cite{Hout3D}). We prove unconditional stability of the scheme, second order of approximation in space and time and positivity of the solution. The results of our numerical experiments demonstrate the above conclusions.

To the best of author's knowledge both approaches have not been considered yet in the literature, so the main results of the paper are new.

The model in use is rather general, in a sense that if considers two (or even three) CEV processes for all the diffusion components and a wide class of the L{\'e}vy processes for the jump components. Therefore, a stable, accurate and sufficiently fast finite-difference approach for pricing derivatives using this model, which is proposed in this paper, could be beneficial for practitioners.

\section*{Acknowledgments}
We thank Peter Carr, Alex Lipton and Alex Veygman for their useful comments and discussions.
Also comments and suggestions of two anonymous referees are highly appreciated.
We assume full responsibility for any remaining errors.


\newcommand{\noopsort}[1]{} \newcommand{\printfirst}[2]{#1}
  \newcommand{\singleletter}[1]{#1} \newcommand{\switchargs}[2]{#2#1}


\appendix
\section{Proof of Proposition \protect\ref{propPos} \label{appPos}}
Recall that for a positive correlation $\rho_{s,v} \ge 0$ we want to prove that the finite-difference scheme:
\begin{align} \label{ap1}
\Big(Q I_v + \sdt W(v) A^B_{2,v}\Big) V^* &= \alpha^{+}V(\tau) - V^k, \\
\Big(P I_S - \sdt \rho_{s,v} W(S) A^F_{2,S} \Big)V^{k+1} &= V^*. \nonumber	\\
\alpha^{+} = (PQ+1) I &- Q\sdt \rho_{s,v} W(S) A_{S}^B + P\sdt W(v) A_{v}^F. \nonumber
\end{align}
\noindent is unconditionally stable in time step $\Delta \tau$, approximates \eqref{spl2} with $O(\sqrt{\dtau}\max(h_S,h_v))$ and preserves positivity of the vector $V(x,\tau)$ if $Q = \beta \sdt / h_v, P = \beta \sdt/h_S$, where $h_v, h_S$ are the grid space steps correspondingly in $v$ and $S$ directions, and the coefficient $\beta$ must be chosen to obey the condition:
\[
\beta > \max_{S,v}[W(v) + \rho_{s,v} W(S)].
\]

First, let us show how to transform \eqref{trickNeg} to \eqref{ap1}. Observe, that \eqref{trickNeg} can be re-written in the form
\begin{align} \label{trickNeg1}
\Big(P &- \sdt \rho_{s,v} W(S)\triangledown_S \Big)\Big(Q + \sdt W(v) \triangledown_v\Big) V^{(1)}
= V(\tau) - V^{(1)} + \alpha V^{(1)} \nonumber \\
&= (\alpha + 1)V(\tau) - V^{(1)} + \alpha[V^{(1)} - V(\tau)], \\
\alpha &= PQ - Q\sdt \rho_{s,v} W(S) \triangledown_S + P\sdt W(v) \triangledown_v. \nonumber
\end{align}
According to \eqref{expSpl}, $V^{(1)} - V(\tau) = \Delta \tau F_{Sv}(\tau) + O\left((\Delta \tau)^2\right)$. Also based on the proposition statement, $P \propto \sdt, \ Q \propto \sdt$, therefore $\alpha[V^{(1)} - V(\tau)] = O\left((\Delta \tau)^2\right)$. As we need just the first order approximation of \eqref{expSpl}, this term in \eqref{trickNeg1} can be omitted. This gives rise to \eqref{ap1}.

\paragraph{Non-negativity} Now prove non-negativity of the solution. Consider first iteration at $k=0$, so the rhs of the first line in \eqref{ap1} can be written as $M_R V^0$, where
\[
M_R \equiv PQ - Q\sdt \rho_{s,v} W(S) A_{S}^B + P\sdt W(v) A_{v}^F,
\]
\noindent is an upper triangular block matrix of the size $N_1 N_2$.
Based on the definitions of the discrete operators $A^F, A^B$ given right before the Proposition~\ref{propPos}, one can see that matrix $M_R$ has all non-negative elements outside of the main diagonal. The elements at the main diagonal $d_0(M_R)$ read
\[
d_0(M_R) = PQ - \sdt \left[\dfrac{Q \rho_{s,v} W(S)}{h_S} + \dfrac{P W(v)}{h_v}\right],
\]
\noindent and are positive if
\[
PQ > \sdt \left[\dfrac{Q \rho_{s,v} W(S)}{h_S} + \dfrac{P W(v)}{h_v}\right].
\]
This can be easily achieved if we put $Q = \beta \sdt/h_v, \ P = \beta \sdt/h_S$. The coefficient $\beta$ must be chosen to obey the condition:
\begin{equation} \label{beta2}
\beta > \max_{S,v}[W(v) + \rho_{s,v} W(S)] > 0,
\end{equation}
\noindent which guarantees that $d_0(M_R) > 0$. Since we require that in the proposition statement, the rhs of \eqref{ap1} is a non-negative vector.

To prove the non-negativity of the solution, consider first the second line in \eqref{ap1}. We need to show that the matrix $M_R^S \equiv P I_S - \sdt \rho_{s,v} W(S) A^F_{2,S}$ is an EM-matrix, see Appendix A in \cite{Itkin2014a}. This can be done similar to the proof of Lemma A.2 in \cite{Itkin2014a}, if one observes that the diagonal elements of $M_R^S$ are positive, i.e.
\begin{equation} \label{dii}
d_{i,i}(M_R^S) = \dfrac{\sdt}{h_S}\left( \beta + \dfrac{3}{2}\rho_{s,v} W(S_i)\right) > 0, \qquad i=1,...,N_1.
\end{equation}
Since $M_R^S$ is an EM-matrix, its inverse is a non-negative matrix, therefore a product of a non-negative matrix and a non-negative vector results in a non-negative vector. Therefore, the non-negativity of the solution is proved.

\paragraph{Convergence of iterations} Since $M_R^S$ is an EM-matrix, all its eigenvalues are non-negative. Also, since this is an upper triangular matrix, its eigenvalues are $ d_{i,i}(M^S_R), \ i=1,...,N_1$. Also, by the properties of an EM-matrix, matrix $\parallel (M_R^S)^{-1} \parallel$ is non-negative, with the eigenvalues $\lambda_i = 1/d_{i,i}(M^S_R), \ i \in [1,...,N_1]$.

Now due to \eqref{beta2} introduce a coefficient $\gamma > 1$ such that
\begin{equation} \label{gamma}
\beta = \gamma \frac{3}{2}\max_{S,v}[W(v) + \rho_{s,v} W(S)].
\end{equation}
From \eqref{dii} we have
\[ d_{i,i}(M_R^S)  > \dfrac{3}{2}\rho_{s,v} W(S_i) \dfrac{\sdt}{h_S}\left( \gamma +1  \right). \]
Thus, it is always possible to provide the condition $d_{i,i}(M_R^S) > 1$ by an appropriate choice of $\gamma$. Accordingly, this gives rise to the condition $|\lambda_i| < 1, \ \forall i \in [1,...,N]$. The latter means that the spectral norm $\parallel (M_R^S)^{-1} \parallel < 1$, and, thus, the map $V^k \rightarrow V^{k+1}$ is contractual. This is the sufficient condition for the Picard iterations in \eqref{ap1} to converge. Unconditional stability follows. Other details about EM-matrices and necessary lemmas again can be found in \cite{Itkin2014a}.

For the first line of \eqref{ap1} we claim the same statement, i.e., that the matrix $M_R^v$ is an EM-matrix. The main diagonal elements of $M_R^v$ are also positive, namely
\[
d_{j,j}(M_R^v) = \dfrac{\sdt}{h_v}\left( \beta + \dfrac{3}{2}W(v_j)\right) > 0, \qquad j=1,...,N_2.
\]
The remaining proof again can be done based on definitions and Lemma A.2 in \cite{Itkin2014a}.

Thus, based on these two steps the coefficient $\gamma$ has to be chosen to provide
\[ \dfrac{3}{2}\max_{S,v}[W(v) + \rho_{s,v} W(S)] \dfrac{\sdt}{h_S}\left( \gamma +1  \right) > 1. \]

Since both steps on \eqref{ap1} converge in the spectral norm, and are unconditionally stable,
the unconditional stability and convergence of the whole scheme follows. It also follows that the whole scheme preserves non-negativity of the solution.

\paragraph{Spatial approximation} In \eqref{ap1} the lhs is approximated with the second order in $h_S$, while the first line in the rhs part uses the first order approximation of the first derivative. As $\triangledown_S = A_S^B + O(h_S)$, and in the first line of the rhs of \eqref{ap1} we have a product $\sdt \triangledown_S$, the order of the ignored terms is $O(\sdt h_S)$. So, rigorously speaking, the whole scheme \eqref{ap1} provides this order of approximation.

\section{Proof of Proposition \protect\ref{propPos2} \label{appPos2}}
As compared with the Proposition~\ref{propPos}, this scheme has the only modification. Namely, instead of the first step in \eqref{ap1}
\begin{align*}
\Big(Q I_v + \sdt W(v) A^B_{2,v}\Big) V^* &= \alpha^{+}V(\tau) - V^k, \\
\alpha^{+} = (PQ+1) I &- Q\sdt \rho_{s,v} W(S) A_{S}^B + P\sdt W(v) A_{v}^F. \nonumber
\end{align*}
\noindent we now use
\begin{align} \label{ap13}
\Big(Q I_v + \sdt W(v) A^B_{2,v}\Big) V^* &= \alpha_2^+ V(\tau) - V^k, \\
\alpha^{+}_2 = (PQ+1) I &- Q\sdt \rho_{s,v} W(S) A_{2,S}^B + P\sdt W(v) A_{2,v}^F. \nonumber
\end{align}
Below we want to prove that $V = \alpha_2^+ V(\tau) - V^k$ is a positive vector.

Suppose we start iterations with $V^{0} = V(\tau)$. Then
\begin{align} \label{rhs1}
V &= V_1 V(\tau) + V_2 V(\tau), \\
V_1 &= \dfrac{1}{2}PQ - Q\sdt \rho_{s,v} W(S) A_{2,S}^B, \qquad
V_2 = \dfrac{1}{2}PQ + P\sdt W(v) A_{2,v}^F. \nonumber
\end{align}
The vector $U_1 = V_1 V(\tau)$ can be computed in a loop on $v_i, i=1,...,N_2$. In other words, if $V(\tau)$ is a $N_1 \times N_2$ matrix where the rows represent the $S$ coordinate and the columns represent the $v$ coordinate, each $j$-th column of $U_1$ is a product of matrix $V_1$ and the $j$-th column of $V(\tau)$. By analogy, the vector $U_2 = V_2 V(\tau)$ can be computed in a loop on $S_i, \ i=1,...,N_1$.

Matrix $V_1$ is lower tridiagonal matrix with positive main and first diagonal and negative second diagonal. For instance, on a uniform grid the elements of these diagonals are
\begin{align*}
d_{0,i} &= \dfrac{1}{2}PQ - \dfrac{3}{2 h_S} Q\sdt \rho_{s,v} \sigma(S_i) S_i^c, \\
d_{-1,i} &= \dfrac{4}{2 h_S} Q\sdt \rho_{s,v} \sigma(S_i) S_i^c, \qquad
d_{-2,i} = -\dfrac{1}{2 h_S} Q\sdt \rho_{s,v} \sigma(S_i) S_i^c. \nonumber
\end{align*}
$V_1$ can be made highly diagonal dominant by choosing $\gamma$ in \eqref{gamma} big enough. On the other hand, each column of $V(\tau)$ is a vector of the option prices on the grid in $S$, i.e. we expect it to be relatively smooth. Therefore, the product of this column with $V_1$ is expected to be positive.

One can also regard to the following intuition. As far as $U_2$ is concerned, it can be observed that $A_{2,v}^F$ in the expression \eqref{rhs1} for $V_2$ is an option Vega, so it has to be positive up to some computational errors. The first term in $V_2$ is big enough to compensate for this possible negative errors. For $U_1$ observe that with allowance for the statement of this Proposition, matrix $V_1$ can be represented as
\[ V_1 > Q \dfrac{\sqrt{\dtau}}{h_S} \rho_{s,v} W(S) \left[ \dfrac{3}{4}\gamma - h_S A_{2,S}^B \right]. \]
Therefore, by taking $\gamma > 4/3$ we get
\[ U_1 > Q \dfrac{\sqrt{\dtau}}{h_S} \rho_{s,v} W(S) \left[ V(\tau) - h_S A_{2,S}^B V(\tau) \right]
\approx Q \dfrac{\sqrt{\dtau}}{h_S} \rho_{s,v} W(S) \left[C - h_S \fp{C}{S}\right] > 0, \]
\noindent with $C$ being the  option price.

Since $V_1$ is a lower tridiagonal matrix, the complexity of computing $V_1 V(\tau) $
is linear in $N_1 N_2$.

Same arguments can be applied to the product of $V_2$ and rows of $V(\tau)$.

The whole FD scheme in the Proposition~\ref{propPos2} in addition to \eqref{ap13} also includes the second step which is same as in the Proposition~\ref{propPos}. Since this step has the second order approximation in spatial variables, the whole scheme also provides the second order. The prove of convergence of the whole scheme is similar to that in Proposition~\ref{propPos}.

\end{document}